\documentclass[10pt, conference, letterpaper]{IEEEtran}
\IEEEoverridecommandlockouts
\usepackage{amsmath}
\usepackage{amsthm}
\usepackage{amsfonts}
\usepackage{algorithmic}
\usepackage{algorithm}
\usepackage{graphicx}
\usepackage{xcolor}
\usepackage{subcaption}
\usepackage{comment}
\usepackage{multirow}
\usepackage{verbatim}
\usepackage{dsfont}

\newcommand{\revise}[1]{{#1}}

\newtheorem{lemma}{Lemma}
\newtheorem{theorem}{Theorem}

\def\eg{\textit{e.g.}}
\def\ie{\textit{i.e.}}
\def\etal{\textit{et al.}}

\begin{document}
\title{Efficient Pipeline Planning for Expedited Distributed DNN Training
\thanks{This work was supported in part by Alibaba Group through Alibaba Innovative Research (AIR) Program, and grants from Hong Kong RGC under the contracts HKU 17204619, 17208920 and 17207621.}
}
\author{\IEEEauthorblockN{Ziyue Luo\IEEEauthorrefmark{1}, Xiaodong Yi\IEEEauthorrefmark{1}, Guoping Long\IEEEauthorrefmark{2}, Shiqing Fan\IEEEauthorrefmark{2}, Chuan Wu\IEEEauthorrefmark{1}, Jun Yang\IEEEauthorrefmark{2}, Wei Lin\IEEEauthorrefmark{2}} 
	\IEEEauthorblockA{\IEEEauthorrefmark{1}Department of Computer Science, The University of Hong Kong, Email: \{zyluo, xdyi, cwu\}@cs.hku.hk}
	\IEEEauthorblockA{\IEEEauthorrefmark{2}Alibaba, Email: longguoping@gmail.com, \{shiqing.fsq, muzhuo.yj, weilin.lw\}@alibaba-inc.com}}
\maketitle

\begin{abstract}
	 To train modern large DNN models, pipeline parallelism has recently emerged, which distributes the model across GPUs and enables different devices to process different microbatches in pipeline. Earlier pipeline designs allow multiple versions of model parameters to co-exist (similar to asynchronous training), and cannot ensure the same model convergence and accuracy performance as without pipelining. Synchronous 
	 pipelining has recently been proposed which ensures model performance by enforcing a synchronization barrier between training iterations. 
	 Nonetheless, the synchronization barrier requires waiting for gradient aggregation from all microbatches and thus delays the training progress. Optimized pipeline planning is needed to minimize such wait and hence the training time, which has not been well studied in the literature. This paper designs efficient, near-optimal algorithms for expediting synchronous pipeline-parallel training of modern large DNNs over arbitrary inter-GPU connectivity. Our algorithm framework comprises two components: a pipeline partition and device mapping algorithm, 
	 and a pipeline scheduler that decides processing order of microbatches over the partitions, which together minimize the per-iteration training time. We conduct thorough theoretical analysis, extensive testbed experiments and trace-driven simulation, and demonstrate our scheme can accelerate training up to 157\% compared with state-of-the-art designs. 
\end{abstract}

\section{Introduction}

Deep learning has advanced various applications in a wide range of domains \cite{devlin2018bert}\cite{he2016deep}\cite{cully2015robots}. 
Deep Neural Networks (DNNs) have significantly grown in size in recent years, in pursuit of better model accuracy. 
Training of large models over large datasets 
has promoted the rapid development of distributed DNN training frameworks (e.g., TensorFlow~\cite{abadi2016tensorflow}, PyTorch~\cite{adam2019pytorch}). 

A number of parallel-training paradigms have been adopted in practice.
Data parallelism~\cite{li2014scaling} partitions the training dataset among workers.  
Each worker holds a 
copy of the DNN, computes parameter updates using the local dataset, and synchronizes parameter updates with others periodically. 
AllReduce operation~\cite{sergeev2018horovod} is a common approach for parameter synchronization among workers. 
To handle large models which cannot be fit entirely into a single device's memory, model parallelism~\cite{shoeybi2019megatron} partitions a DNN model and places model partitions on different devices. In each training iteration, a mini-batch is processed by model partitions on the devices one after another, through forward propagation followed by backward propagation. 
Such vanilla model parallelism suffers from low device utilization, as only one device is active at each time when a mini-batch is trained 
across the devices hosting different model partitions. 
Pipeline parallelism~\cite{harlap2016addressing} has been proposed 
to maximize device utilization during training. 
Similar to model parallelism, it partitions a DNN into stages and places stages over multiple devices; 
it further partitions each mini-batch of training data into equal-sized microbatches, and 
allows different devices to process different microbatches at the same time (\ie, microbatch pipelining). 

Most works on pipeline parallelism~\cite{harlap2018pipedream}\cite{geng2019elasticpipe}\cite{Park2020hetpipe}\cite{pmlr-v139-narayanan21a} adopt asynchronous pipelining, by injecting microbatches into the training pipeline one by one and updating model parameters with gradients computed with a microbatch, whenever its backward propagation is done. 
Asynchronous pipeline parallelism maximizes GPU utilization by fully saturating the pipeline. 
However, the processing of different microbatches overlaps, each updating the model using gradients computed based on outdated parameters that are learned on different earlier microbatches, which
may inevitably slow down or even prevent training convergence, and render a model whose accuracy differs from that trained without pipelining~\cite{ho2013more}. 

To ensure model convergence and accuracy, synchronous pipeline parallelism has been advocated by a few recent studies~\cite{huang2019gpipe}\cite{fan2021dapple}. It enforces a synchronization barrier between training iterations, to aggregate gradients computed with all microbatches before applying them for model update. Such a synchronization barrier flushes the pipeline and introduces waiting time (for training completion of all microbatches) into each training iteration, leading to lower device utilization as compared to asynchronous pipeline training. Optimal planning of synchronous pipeline training is needed to 
improve device utilization and minimize per-iteration training time, to achieve similar training time as asynchronous pipelining while providing convergence and accuracy guarantees. Pipeline planning includes DNN model partition, replication and device placement, as well as scheduling the order of microbatch processing across the devices within each training iteration. 
 Non-trivial challenges exist, 
 as follows: 



\textit{First}, in a typical DNN model, layers are not uniform in terms of computation time, parameter size and activation size. Optimal model partition over devices is hence challenging. 

\textit{Second}, previous pipeline designs have been restricted to homogeneous inter-GPU connectivity (or homogeneous in each level of a hierarchical topology)~\cite{huang2019gpipe}\cite{harlap2018pipedream}. GPU inter-connectivity is often more complicated in a practical machine learning (ML) cluster, 
including PCI-e or NVLink within a physical server~\cite{dgx1}, RDMA or TCP network between servers~\cite{wang2013gpu}. 
We will show that heterogeneous GPU inter-connectivity leads to an exponential number of solutions for DNN model partition and device mapping (Sec.~\ref{sec::part}), adding to the difficulty of finding efficient, near-optimal solutions.

\textit{Third}, deciding the execution order of all microbatches over all devices, 
respecting inter-stage dependencies and minimizing per-iteration training time, 
falls in the category of job shop problems. 
The job shop problem is NP-hard~\cite{goldberg2001better} even with only two machines (aka GPUs in our case). 

Tackling the challenges, we design near-optimal algorithms 
that efficiently partition a given DNN model, replicate and distribute the partitions over available GPUs with arbitrary inter-GPU connectivity, and schedule microbatch processing over the stages to minimize per-iteration training time.
Our main techniques and contributions are summarized as follows:

$\triangleright$ Assuming model partition and device mapping are given, we 
design an efficient list ordering method to decide the processing order of microbatches on different GPUs, and then a scheduling algorithm that 
minimizes idle time of devices based on the order. With thorough theoretical analysis, we identify an upper bound of per-iteration training time, decided by two key factors: the number of stages that the DNN is partitioned into, and the maximum time to process a microbatch on a single stage or an inter-stage communication channel.


$\triangleright$ We are hence inspired to design a pipeline partition and device mapping algorithm to minimize the 
 maximum per-stage/channel execution time, given the number of stages to partition the model into. 
A recursive min-cut method is designed to identify a device order that maximizes inter-GPU bandwidth utilization. 
Based on the device order, we use dynamic programming to derive the optimal partition, replication and mapping
 solution. 

 $\triangleright$ 
Our complete synchronous pipeline planning algorithm, {\em SPP}, iteratively invokes the pipeline partition/mapping algorithm and the execution order scheduler to identify the best number of partitioned stages and the set of near-optimal pipeline execution strategies accordingly. 
We rigorously analyze {\em SPP} and prove a worst-case performance guarantee.
 
 $\triangleright$ We conduct extensive testbed experiments and trace-driven simulation, carefully comparing {\em SPP} with state-of-the-art pipeline-training paradigms, including 
 GPipe~\cite{huang2019gpipe}, PipeDream~\cite{harlap2018pipedream} and HetPipe~\cite{Park2020hetpipe}.
  Experimental results show that  {\em SPP} can accelerate training up to 147\% compared to GPipe, 157\% to PipeDream and 80\% to HetPipe in terms of per-iteration training time, 
  and achieves the target accuracy in the most expedited manner as compared to baselines. 
  We observe that {\em SPP} can strike a balance between the number of stages and the maximum per-stage execution/communication time in DNN partition, and maximally overlap communication and computation with its pipeline execution scheduling. 


\vspace{-2mm}
\section{Background and Related Work}
\label{sec::related_work}

\noindent{\bf DNN Training.}
A DNN model comprised of multiple layers is usually trained over a large dataset iteratively to minimize a loss function~\cite{goodfellow2016deep}. The dataset is typically divided into equal-sized mini-batches. In each training iteration, one mini-batch is processed to update the DNN model as follows: (1) {\em forward propagation (FP)}: the mini-batch is computed by each layer of the DNN sequentially to derive a loss; 2) {\em backward propagation (BP)}: gradients of model parameters are computed based on the loss from the last layer to the first layer; 3) 
a {\em gradient update} operation applies computed gradients to parameters in each layer with an optimization algorithm, \eg,~stochastic gradient descent (SGD) or adaptive moment estimation (Adam)~\cite{goodfellow2016deep}.

\vspace{-2mm}
\noindent{\bf DNN Model Partition and Device Mapping.}
A number of studies have focused on partition and device mapping strategies for large DNN models through learning-based methods~\cite{mirhoseini2017device}\cite{addanki2018placeto}\cite{yi2020optimizing}, which require large computing resources and long training time to derive a satisfying policy for one training job.
A few efforts~\cite{wu2020stanza}\cite{yi2020fast} exploit efficient heuristics for DNN model partition and device mapping at the operation level, requiring detailed cost modeling of the DNN model and accurate profiling of operation execution time.
Our work focuses on layer-level DNN model partitioning and mapping, and derives a polynomial-time pipeline planning strategy.

\vspace{-2mm}
\noindent{\bf Data Parallelism (DP) and Model Parallelism (MP)}
are commonly adopted to parallelize training across multiple devices. As shown in Fig.~\ref{fig_mp_pp}(a), with DP, 
three mini-batches 
are each trained on one GPU with a complete copy of model parameters; an AllReduce operation synchronizes computed gradients after training of all mini-batches.
With MP (Fig.~\ref{fig_mp_pp}(b)), in each training iteration, a mini-batch is fed into the device hosting the first stage(s) of the DNN model for FP, and the computed 
activations are passed to later stages on other devices for FP; during BP, gradients are computed and passed from one device to another following reverse sequence of the stage(s). In this way, only one device is active at each time, where \revise{FP} or BP of the mini-batch is being carried out, while other devices are idle, leading to low device utilization.

\begin{figure}[!t]
\centering
    \begin{subfigure}{0.89\columnwidth}
      \centerline{\includegraphics[width=\columnwidth]{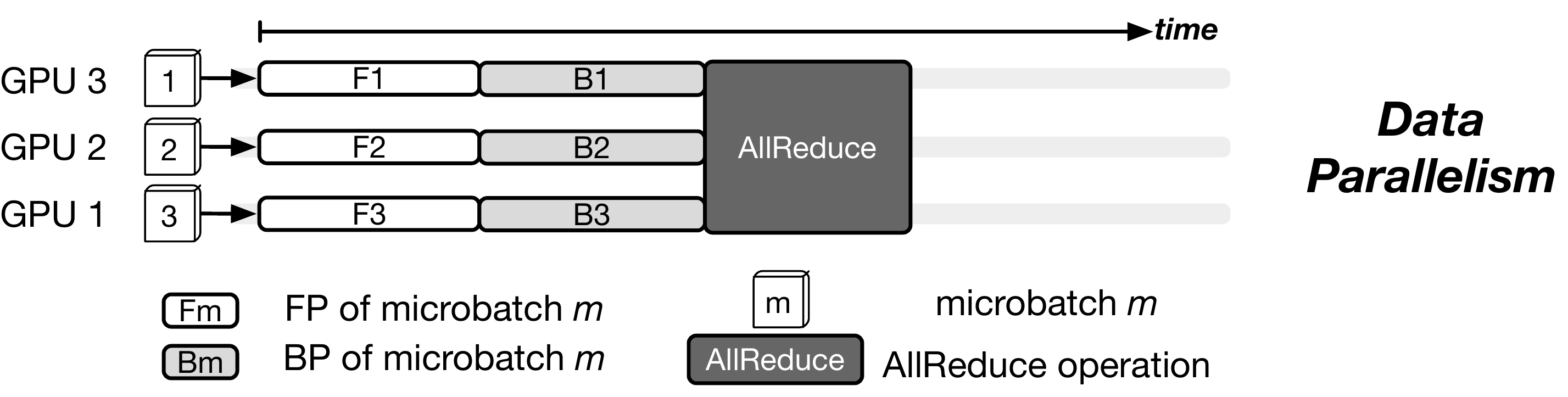}}
      \caption{Data Parallelism}
    \end{subfigure}
    \begin{subfigure}{0.89\columnwidth}
      \centerline{\includegraphics[width=\columnwidth]{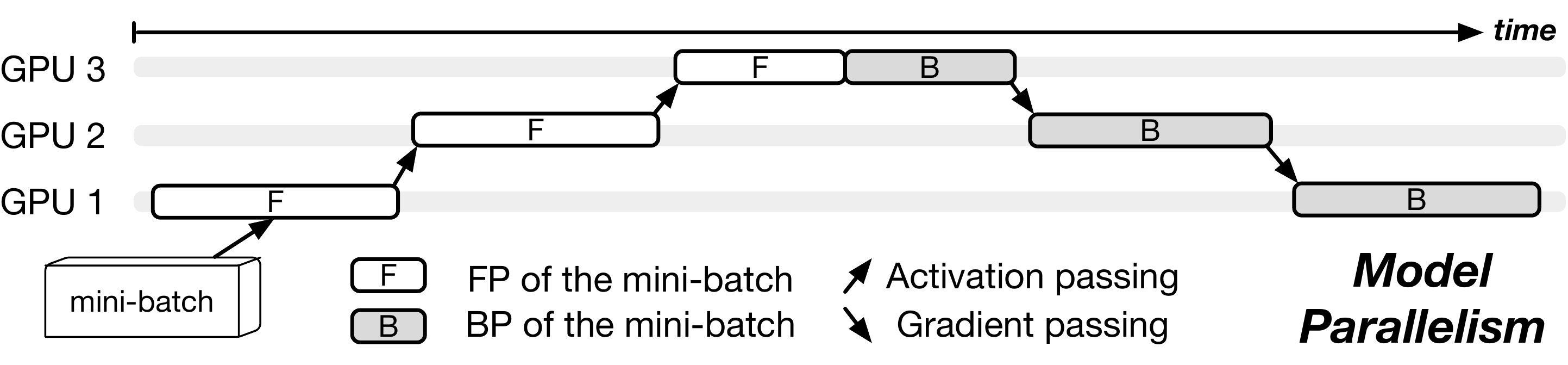}}
      \caption{Model Parallelism}
    \end{subfigure}
    \begin{subfigure}{0.89\columnwidth}
      \centerline{\includegraphics[width=\columnwidth]{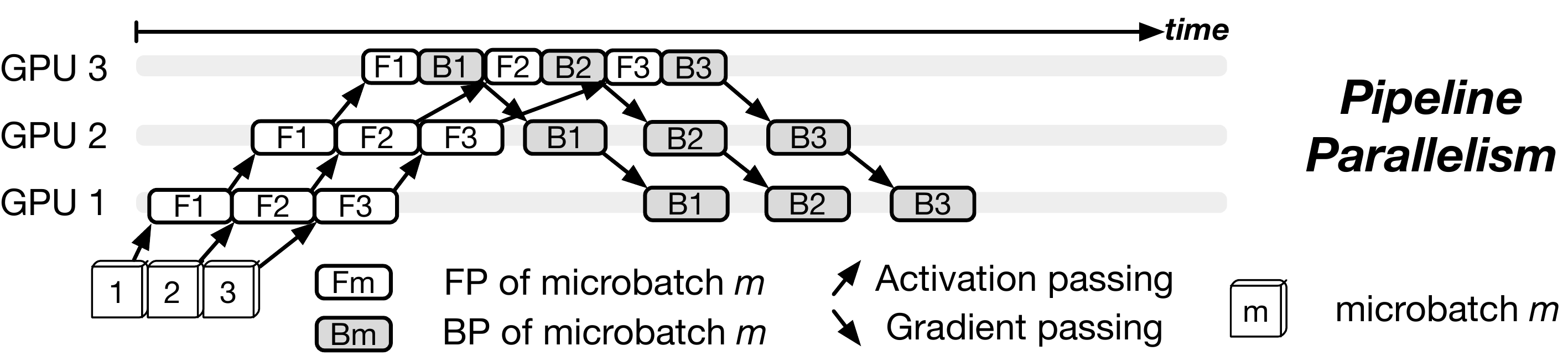}}
      \caption{Pipeline Parallelism}
    \end{subfigure}
    \caption{Data parallelism vs. model parallelism vs. pipeline parallelism: 
    1 mini-batch divided into 3 microbatches.}
    \label{fig_mp_pp}
\end{figure}
\setlength{\textfloatsep}{0pt}

\vspace{-2mm}
\noindent\textbf{Pipeline Parallelism}
\label{ppparallelism}
Based on model parallelism, \textit{pipeline parallelism} further divides a mini-batch into equal-sized microbatches (Fig.~\ref{fig_mp_pp}(c)).
The microbatches 
are consecutively fed into the device hosting the first stage(s) whenever the forward computation of the previous microbatch is done on this device, rendering a training pipeline. Consequently, it enables multiple devices to 
process different microbatches simultaneously. 

{\em Asynchronous Pipeline Training.} PipeDream~\cite{harlap2018pipedream} partitions a DNN model over multiple servers, allowing stage replication among the servers, and further divides a stage over GPUs within each server, aiming at minimizing the maximum time to process a single stage. Its server configuration and inter-server connectivity are both homogeneous. Stage execution is scheduled to ensure that every FP stage is immediately followed by a BP stage. 
Geng \etal~\cite{geng2019elasticpipe} study pipeline parallelism over heterogeneous GPUs, and propose a dynamic tuning algorithm to identify straggler devices and redistribute the DNN model for better load balance. 
In HetPipe~\cite{Park2020hetpipe}, each node (comprised of homogeneous GPUs) trains the DNN model in a pipelined manner similar to PipeDream without stage replication; 
DP is used for training and parameter synchronization among nodes. 
With asynchronous pipelining, microbatches are trained on outdated versions of model parameters to compute gradients, 
leading to slow model convergence 
and lower accuracy of the obtained model as compared to synchronous training~\cite{ho2013more}. 
Several studies have investigated mitigating the accuracy loss 
via weight prediction~\cite{chen2018efficient} or randomized smoothing~\cite{colin2019theoretical}, under restricted assumptions of training loss functions.  

{\em Synchronous Pipeline Training.} 
GPipe~\cite{huang2019gpipe} is a synchronous pipeline training framework, 
including (1) a partition strategy that ensures approximately the same number of DNN layers on each GPU, and (2) a schedule to execute all FP before starting any BP. It does not allow stage replication and provides no device mapping strategies.
We design efficient algorithms to deal with all aspects of synchronous pipeline planning. 
\section{System Model}
\label{sec::sys}


\subsection{DNN Model and Device Connectivity}

We consider a DNN model, $\mathcal{D}$, consisting of $L$ layers, \eg, attention layers, convolutional and fully-connected layers. In each training iteration, a mini-batch is divided into $M$ equal-sized microbatches of size $Z$ each. Every microbatch is trained through FP through all $L$ layers, followed by BP over the $L$ layers in the reverse order. We divide $\mathcal{D}$ into multiple {\em stages}, place the stages on different GPUs, and allow different GPUs to process different microbatches simultaneously in the pipelined manner. Following the end of BP of all microbatches, 
a gradient aggregation operation aggregates gradients computed from all microbatches 
and applies 
them to update the model parameters.
As time needed for gradient aggregation and apply is much shorter than FP/BP time , we ignore it in our pipeline parallelism design.

$V$ homogeneous GPUs on multiple physical servers are available for training this DNN model.\footnote{Training a DNN using GPUs of the same type is the norm in today's production systems, based on our discussions with leading AI cloud operators.} We 
consider a variety of GPU inter-connectivity, 
including PCIe or NVLink 
{(providing direct GPU-GPU communication channel)} between GPUs
in the same physical server (\eg, in NVIDIA DGX-1~\cite{dgx1}), TCP or RDMA connections across GPUs in different servers~\cite{wang2013gpu}, 
 with various bandwidth levels. 
Graph $G(\mathcal{V}, \mathcal{E})$ represents the multi-GPU system for training $\mathcal{D}$, 
where $\mathcal{V}$ includes the $V$ GPUs 
and $\mathcal{E}$ contains all inter-connections between the GPUs. \revise{Each edge $(v, v')$ in $\mathcal{E}$ is associated with a weight, $b_{v, v'}$, representing the available bandwidth between GPU $v$ and GPU $v'$.} Let $b_{min}$ and $b_{max}$ be the minimum and maximum bandwidth among all edges in $\mathcal{E}$, respectively.

The forward (backward) computation time of a microbatch over layer $l$ of DNN $\mathcal{D}$ on a given GPU is $p^f_l$ ($p^b_l$). Let $\alpha_l$ be the size of parameters/gradients of layer $l$, which can be profiled through one trial run of the whole model using several training iterations.
$d^f_{l,l+1}$ denotes the size of activations passed from layer $l$ to layer $l+1$ during FP, and $d^b_{l+1,l}$ is the size of gradients transferred from layer $l+1$ to layer $l$ during BP.

\begin{figure*}[!th]
\centering
    \begin{subfigure}{0.7\columnwidth}
      \centerline{\includegraphics[width=0.9\columnwidth]{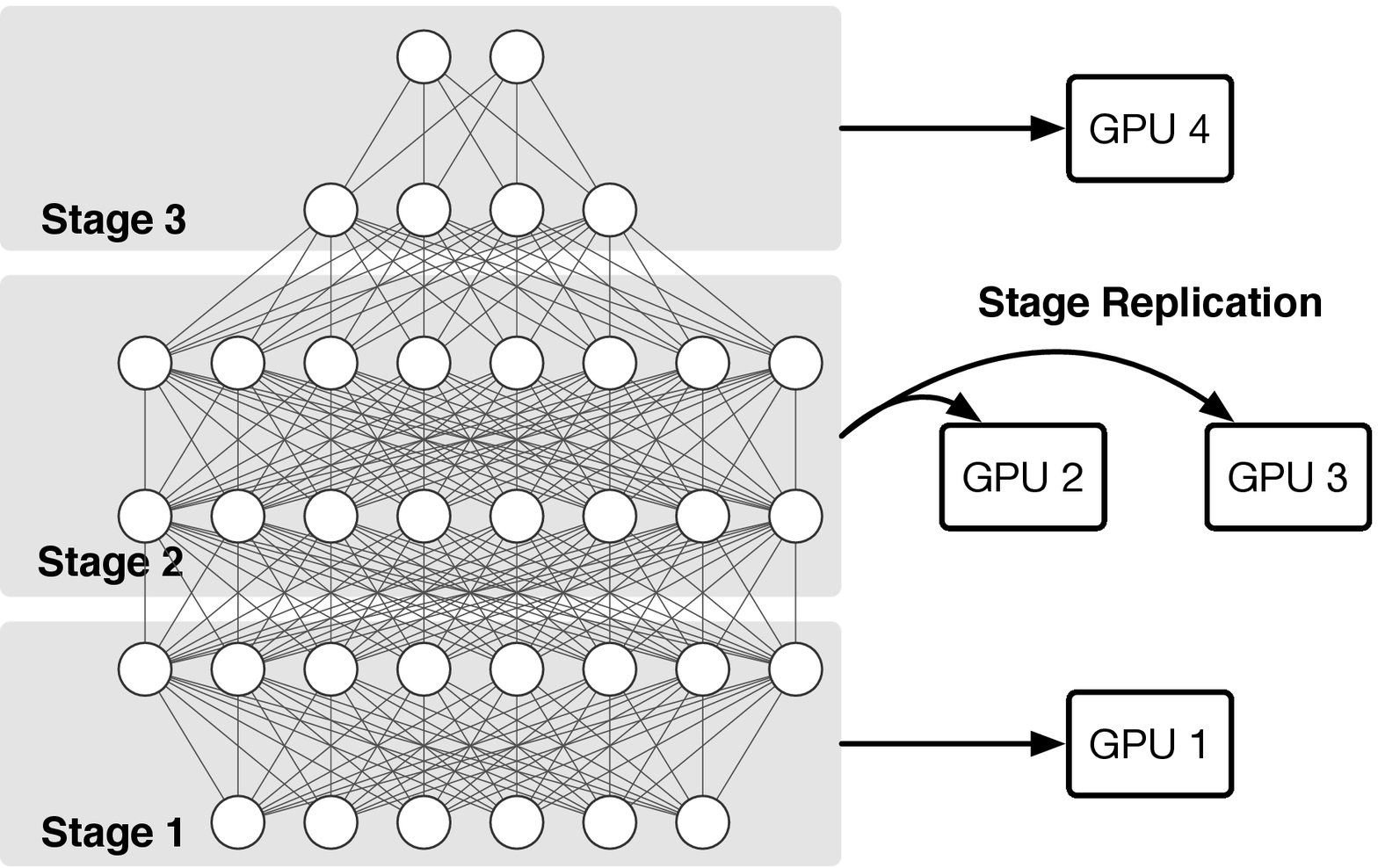}}
      \caption{Stage partition, replication and device mapping}
    \end{subfigure}
    \begin{subfigure}{1.0\columnwidth}
      \centerline{\includegraphics[width=1\columnwidth]{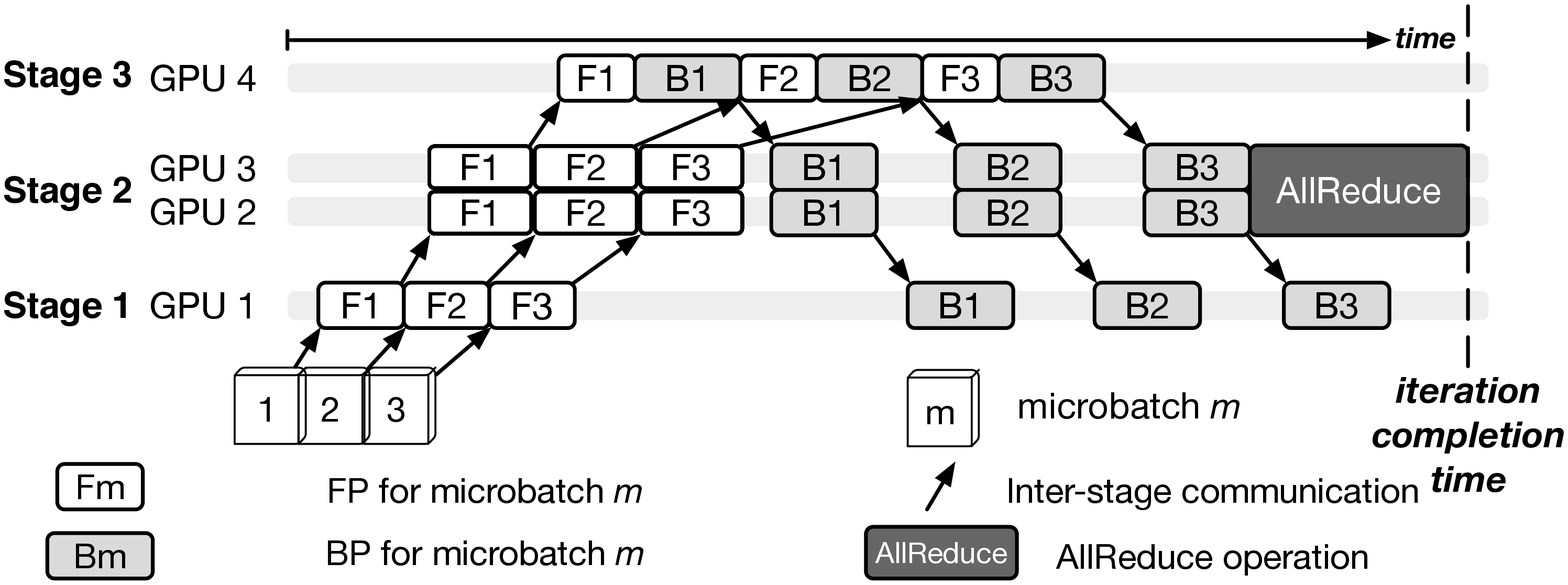}}
      \caption{Execution schedule of 3 microbatches}
    \end{subfigure}
    \caption{A pipeline parallelism design example}
    \vspace{-5mm}
    \label{fig_pipeline_planner}
\end{figure*}
\setlength{\textfloatsep}{0pt}

\vspace{-2mm}
\subsection{Synchronous Pipeline with Stage Replication}

We target synchronous pipeline design that minimizes per-iteration training time, 
including two subproblems:

\subsubsection{\bf Pipeline Partition, Replication and Device Mapping} 
\label{sec:partmapproblem}

We decide the stage set $\mathcal{S} = \{s_1, s_2, \ldots, s_{|\mathcal{S}|}\}$ with $|\mathcal{S}| \leq V$ to partition model $\mathcal{D}$
into, and a device mapping function $\mathcal{F}: \mathcal{S} \rightarrow \mathds{P}(\mathcal{V})$, where $\mathds{P}(\mathcal{V})$ 
includes all subsets of device set $\mathcal{V}$.

We consider classical \textit{interval partition}, 
that each 
stage consists of a number of consecutive layers: 
for stage $s \in \mathcal{S}$, if layer $l_{start}$ and layer $l_{end}$ belong to $s$, then all layers $l$, with $l_{start} \leq l \leq l_{end}$, belong to $s$. Without loss of generality, we assume a sequential dependency through $s_1, s_2, \ldots, s_{|\mathcal{S}|}$, \ie, the last layer $l^n_{end}$ in stage $s_{n}$ is the predecessor of the first layer $l^{n+1}_{start}$ in stage $s_{n+1}$ in the DNN model.

$\mathcal{F}$ maps each stage $s \in \mathcal{S}$ to one or multiple GPUs, ensuring that each GPU hosts exactly one stage or one replica of a stage. 
In our design, 
we allow a stage to be replicated and executed over multiple GPUs in a data-parallel fashion. Suppose stage $s$ is replicated over $k$ GPUs $\{v_1, v_2, \ldots, v_k\}$. 
Processing of a microbatch by stage $s$ is distributed over the $k$ GPUs (by evenly dividing input data among these GPUs), and 
{we assume that the forward (backward) computation time of each layer $l$ in stage $s$ on 
each replica device is now {\small $\frac{p^f_l}{k}  (\frac{p^b_l}{k}$)}.\footnote{We note that non-linear change of training time may happen when a layer is replicated, i.e., each layer replica's execution time is not exactly $\frac{1}{k}$ of the layer's processing time without input data partition. Our algorithm can be readily extended to the non-linear case by modeling the computation time of each layer given different input data sizes.} }
Fig.~\ref{fig_pipeline_planner}(a) gives an example, where a 6-layer DNN model is trained using 4 GPUs. 
The model is partitioned into three stages with stage 2 replicated over GPU 2 and GPU 3.

Such stage replication may improve GPU utilization and further balance stage processing time, together with stage partition. 
In Fig.~\ref{fig_pipeline_planner}, supposing the size of layers in stage 2 is much larger than the other stages, replicating stage 2 on two GPUs allows forward/backward computation time of the stage to be similar to others, as shown in Fig.~\ref{fig_pipeline_planner}(b). 
 
{\small
\begin{table}[!t]
\small\
\caption{NOTATION}
\begin{center}
\label{table_notation}
\begin{tabular}{|l|l|}
\hline
$\mathcal{D}$				& the DNN model\\
\hline
$L$							& \# of layers\\
\hline
$M$							& \# of microbatches in one iteration\\
\hline
$V$							& \# of GPUs\\
\hline
$\mathcal{G}(\mathcal{V}, \mathcal{E})$	& GPU inter-connection graph \\
& ($\mathcal{V}$: GPUs; $\mathcal{E}$: inter-GPU connections)\\
\hline
$b_{v, v'}$					& bandwidth between GPU $v$ and GPU $v'$\\
\hline
$b_{min} (b_{max})$			& minimum (maximum) bandwidth in $\mathcal{E}$\\
\hline
$p^f_l (p^b_l)$				& FB (BP) computation time of layer $l$  per \\& microbatch \\
\hline
$\alpha_l$					& size of parameters (gradients) of layer $l$\\
\hline
$d^f_{l,l+1}$ ($d^b_{l+1,l}$)			& size of activations (gradients) from \\& layer $l$ to  $l+1$ ($l+1$ to $l$) during FP (BP)\\
\hline
$\mathcal{S}$				& set of all stages that $\mathcal{D}$ is partitioned into\\
\hline
$\mathcal{S}_{repl}$		& set of replicated stages\\
\hline
$\mathcal{F}: \mathcal{S} \rightarrow \mathds{P}(\mathcal{V})$		& device mapping function from stages to sets\\&  of GPUs\\
\hline
{\small$c^f_{s_n, s_{n+1}}$}	& communication time between stages $s_n$ and\\
{\small $(c^b_{s_{n+1}, s_n})$}& $s_{n+1}$ during FP (BP)\\
\hline
$e^f_{m, s_n} (e^b_{m, s_n})$		& start time of stage $s_n$'s processing of \\& microbatch $m$ during FP (BP)\\
\hline
$A_s$						& time taken by AllReduce operation of stage $s$\\ 
\hline
$e^A_{s}$		& start time of AllReduce operation of stage $s$\\
\hline
\end{tabular}
\end{center}
\vspace{-2mm}
\end{table}
}

After completion of backward computation of microbatches on all $k$ replicas of stage $s$, a ring AllReduce operation~\cite{sergeev2018horovod} synchronizes gradients of stage $s$ across the $k$ GPUs. 
{Especially, the $k$ GPUs are organized into a logical ring topology, 
and each GPU exchanges gradients/parameters with its neighbors in the ring through inter-GPU  connections.}
 The size of communication data (gradients and parameters) involved in the AllReduce operation is $\frac{2(k-1)}{k}\sum_{l\in s}\alpha_l$ per GPU~\cite{allreduce_time}. The time taken by the AllReduce operation, denoted by $A_s$, is further 
 decided by the minimum connection bandwidth among the $k$ GPUs:

\vspace{-4mm}
{\small
\begin{eqnarray}
	\label{eqn_allreduce}
	A_s = \frac{2(k-1)\sum\limits_{l\in s}\alpha_l}{k\min\limits_{v,v'\in \{v_1, v_2, \ldots, v_k\}}b_{v,v'}}
\end{eqnarray}}
\vspace{-3mm}

\subsubsection{\bf Execution Scheduling}
\label{sec:scheduleproblem}

We also decide the execution order of processing each microbatch on each stage, as well as running the AllReduce operations for replicated stages. 
Let $e^f_{m, s_n}$ ($e^b_{m, s_n}$) be the start time of forward (backward) computation of microbatch $m$ 
 on stage $s_n$. 

Execution schedule should respect forward-backward dependency and stage dependency. Each GPU can only 
process one microbatch 
at a time. Let $c^f_{s_n, s_{n+1}}$ and $c^b_{s_{n+1}, s_{n}}$ represent the inter-stage communication time between stage $s_n$ and stage $s_{n+1}$ during FP and BP, respectively. We ignore the time for data passing between layers residing in the same GPU. We formulate the dependencies as follows. 

\begin{itemize}
	\item (\textit{Forward-backward dependency}): 
\end{itemize}
\vspace{-5mm}

{\small 
\begin{eqnarray*}
	e^f_{m, s_{|\mathcal{S}|}} + \frac{\sum\limits_{l \in s_{|\mathcal{S}|}}p^f_l}{|\mathcal{F}(s_{|\mathcal{S}|})|} \leq e^b_{m, s_{|\mathcal{S}|}}, \forall m\in \{1, \ldots, M\}
\end{eqnarray*}
}
\vspace{-5mm}

\begin{itemize}
	\item (\textit{Stage dependency}): 
\end{itemize}
\vspace{-5mm}

{\small
\begin{eqnarray*}
    e^f_{m, s_n} + \frac{\sum\limits_{l \in s_{n}}p^f_l}{|\mathcal{F}(s_{n})|} + c^f_{s_n, s_{n+1}} \leq e^f_{m, s_{n+1}}, \\
    \forall m\in \{1, \ldots, M\}, n\in \{1, \ldots, |\mathcal{S}|-1\}\\
%
	e^b_{m, s_n} + \frac{\sum\limits_{l \in s_{n}}p^b_l}{|\mathcal{F}(s_{n})|} + c^b_{s_n, s_{n-1}} \leq e^b_{m, s_{n-1}}, \\
	\forall m\in \{1, \ldots, M\}, n\in \{2, \ldots, |\mathcal{S}|\}\\
 	e^f_{1, s_1}=0
\end{eqnarray*}
}
\vspace{-6mm}


%

To compute inter-stage communication time, when $s_n$ and/or $s_{n+1}$ are replicated over multiple GPUs, we evenly distribute the data being transmitted across inter-stage links. For example in Fig.~\ref{fig_pipeline_comm}, $s_n$ is replicated onto 2 GPUs and $s_{n+1}$ onto 3 GPUs. 
During FP, 1/3 of the activations derived by each of the two GPUs hosting stage $n$ is sent to each of the three GPUs hosting stage $n+1$. During backward propagation, each GPU running $s_{n+1}$ splits its gradients (computed with a microbatch) into two sets of gradients corresponding to two smaller batches, and sends the two sets to the two GPU of $s_n$, respectively; each replica of $s_n$ sums up received gradients from replicas of $s_{n+1}$. 
The inter-stage communication time is decided by the minimum link bandwidth between GPUs in $\mathcal{F}(s_n)$ and in $\mathcal{F}(s_{n+1})$: 
(note $d^f_{l^n_{end},l^{n+1}_{start}}$ and $d^b_{l^{n+1}_{start}, l^n_{end}}$ are data size produced by an entire microbatch, \ie, sum of data produced by all replicas of a stage)
\vspace{-3mm}
\begin{figure}[!t]
  \centering
  \includegraphics[width=0.35\textwidth]{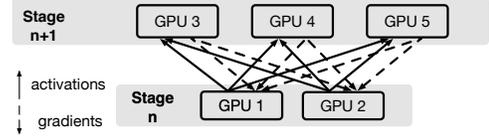}
  \vspace{-2mm}
  \caption{Inter-stage communication: an example}
  \label{fig_pipeline_comm}
\end{figure}
\setlength{\textfloatsep}{0pt}

{\small
\begin{eqnarray*}
	c^f_{s_n, s_{n+1}} = \frac{d^f_{l^n_{end}, l^{n+1}_{start}}}{|\mathcal{F}(s_{n})||\mathcal{F}(s_{n+1})|\min\limits_{v\in\mathcal{F}(s_n), v'\in\mathcal{F}(s_{n+1})}b_{v,v'}}
\end{eqnarray*}
\vspace{-6mm}

\begin{eqnarray*}
	c^b_{s_{n+1}, s_{n}} = \frac{d^b_{l^{n+1}_{start}, l^n_{end}}}{|\mathcal{F}(s_{n})||\mathcal{F}(s_{n+1})|\min\limits_{v\in\mathcal{F}(s_n), v'\in\mathcal{F}(s_{n+1})}b_{v',v}}
\end{eqnarray*} 
}
\vspace{-4mm}

Let $\mathcal{S}_{repl}$ be the set of replicated stages, and $e^A_{s}$ denote the start time of AllReduce operation of stage $s\in \mathcal{S}_{repl}$. We further have:

$\bullet$ (\textit{AllReduce operation dependency}):

\vspace{-2mm}
{\small
\begin{displaymath}
	e^b_{m, s} + \frac{\sum\limits_{l \in s}p^b_l}{|\mathcal{F}(s)|} \leq e^A_{s}, \forall m\in \{1,\ldots, M\}, s\in \mathcal{S}_{repl}
\end{displaymath}
}
\vspace{-3mm}

We aim at minimizing the makespan of training all $M$ microbatches, \ie, the per-iteration training time of the DNN:

\vspace{-4mm}
{\small
\begin{equation}
\mbox{minimize } \max\{\max_{m\in\{1, 2, \ldots, M\}} (e^b_{m, s_1} + \frac{\sum\limits_{l \in s_{1}}p^b_l}{|\mathcal{F}({s_1})|}), \max_{s\in\mathcal{S}_{repl}}(e^A_{s} + A_s)\}
\label{eqn_makespan}
\end{equation}}
\vspace{-4mm}

An AllReduce operation and inter-stage communication do not share the same inter-GPU connections: the former uses links between GPUs hosting replicas of the same stage, while the latter is between GPUs hosting different stages (also recall that one GPU can only host (a replica of) one stage). The AllReduce operation of a replicated stage $s_n$ can take place at the same time as inter-stage communication and backward computation of stages $s_{n-1}, s_{n-2}, \ldots$, as well as AllReduce operations of other replicated stages. Therefore, the completion time of a training iteration in (\ref{eqn_makespan}) is decided by the latest among backward computation completion time of all microbatches over stage $1$ and end time of AllReduce operations of all replicated stages.
Note that our schedule may not ensure that microbatches are processed in the same sequence at each stage; instead, we enforce a synchronization barrier in each training iteration, as represented by the inside $\max$ over all microbatches in (\ref{eqn_makespan}).

An example execution schedule is given in Fig.~\ref{fig_pipeline_planner}(b).  
As stage 2 is replicated over two GPUs, an AllReduce operation happens when backward computation of all three microbatches over both GPUs has been done, ensuring the model parameters on GPU 2 and GPU 3 are updated the same.

\vspace{-2mm}
\section{Pipeline Planning Algorithms}
\vspace{-1mm}
We now design algorithms for efficient synchronous pipeline training. We start with execution scheduler design, assuming model partition and device mapping are given; then we devise the partition and device mapping algorithm that minimizes per-iteration training time together with the execution scheduler.

\vspace{-2mm}
\subsection{
Execution Scheduler}
\label{sec::schedule}
\vspace{-1mm}

Given model partitions $\mathcal{S}$ and device mapping $\mathcal{F}$, our scheduling problem, as presented in Sec.~\ref{sec:scheduleproblem}, is a special case of the 
NP-hard job shop problem~\cite{goldberg2001better}: microbatches 
correspond to jobs of the same type and stages correspond to machines in the job shop problem, 
and the objective is to minimize the total time of executing all jobs. We design an efficient 
{\em pipeline execution} (PE) scheduling algorithm to achieve a proven performance bound.

The PE algorithm contains two modules: 
1) an 
ordering method to decide execution order of microbatches over stages on different GPUs, and 2) an 
algorithm that schedules pipeline execution based on the computed order.

\textit{\textbf{1) 
Execution ordering}}: We define a \textit{computation block} as the forward or backward computation of a stage. 
As backward computation of the last stage $s_{|\mathcal{S}|}$ follows immediately forward computation of $s_{|\mathcal{S}|}$, 
we merge stage $s_{|\mathcal{S}|}$'s forward and backward computation blocks into a single computation block.  
We define the inter-stage communication from $s_n$ to $s_{n+1}$ or from $s_{n+1}$ to $s_n$ to be a \textit{communication block}, including all communication over this {\em communication channel}, \ie, the set of connections from GPU(s) hosting the former stage to GPU(s) hosting the latter stage.
The end-to-end training of every microbatch in a training iteration involves $2|\mathcal{S}|-1$ computation blocks, $2|\mathcal{S}|-2$ communication blocks and $|\mathcal{S}_{repl}|$ AllReduce operations for replicated stages. Let $\mathcal{J}=\{1, 2, \ldots, 4|\mathcal{S}|-3\}$ be the ordered list of all computation and communication blocks, with blocks ordered according to their execution dependencies.

An execution order queue, $U_s$, is maintained for each stage $s \in \mathcal{S}$, containing $\tt{(microbatch\ index, }$ $\tt{block\ number)}$ pairs indicating the order of processing microbatches by forward or backward computation blocks of stage $s$. 



For each block $j \in \mathcal{J}$, we maintain an \textit{available microbatch} queue $Q_j$, 
containing microbatches which have been processed by block $j-1$ but not 
by $j$. 
Initially $Q_1$ includes all microbatches in order of their indices, 
and $Q_j = \emptyset, \forall j \in \mathcal{J}/\{1\}$.

We order microbatch processing over the 
blocks as follows. According to the order of blocks in $\mathcal{J}$, we pop out one microbatch $m$ at the head of a non-empty 
queue $Q_{j}$, and push it to the end of queue $Q_{j+1}$ of the next block (if $j$ is not the last block in $\mathcal{J}$); if 
block $j$ is a computation block of stage $s$, we add $(m, j)$ to 
execution order queue $U_s$. Going through the block list, we identify at most one microbatch to be processed by each block, 
corresponding to microbatches that can be processed about simultaneously.
We loop through the block list repeatedly 
until all available microbatch queues are empty ($Q_{j} = \emptyset, \forall j \in \mathcal{J}$), \ie, end-to-end training of all microbatches is ordered. 

\textit{\textbf{2) Scheduling}}: 
We next 
exploit the execution order queues, $U_s$'s, and schedule a microbatch's processing on a block as soon as it is ready. We start by popping the first (microbatch index, block number) out of queue $U_{s_1}$ of the first stage $s_1$, and process the corresponding microbatch on the respective block. Once a computation block is executed, the successor communication block is immediately run (upon the communication channel becoming idle). 
Upon processing completion of a scheduled computation block of stage $s$ or a communication block which transmits data to stage $s$, 
we examine queue $U_s$:
if the first (microbatch index, block number) in $U_s$ is ready to be executed (\ie, the microbatch has been processed by the precedent block), we pop it out and run it. 
This procedure terminates when $U_s = \emptyset, \forall s \in \mathcal{S}$, \ie, all microbatches have been processed by all computation and communication blocks. 

For each replicated stage $s \in \mathcal{S}_{repl}$, when all microbatches have been processed by backward computation block of this stage, the corresponding AllReduce operation is executed. 

{\small
\begin{algorithm}[!th]

\caption{Pipeline Execution Scheduler - \textbf{\textit{PE}}}
\label{alg_schedule}

\renewcommand{\algorithmicrequire}{\textbf{Input:}}
\renewcommand{\algorithmicensure}{\textbf{Output:}}

\begin{algorithmic}[1]
\REQUIRE $\mathcal{G}(\mathcal{V}, \mathcal{E}), \mathcal{S}, \mathcal{F}: S \rightarrow \mathbb{P}(\mathcal{V})$
\ENSURE $T_{PE}, \textbf{e}$ 

\STATE Initialize execution order queues $U_{s} \leftarrow \emptyset, \forall s \in \mathcal{S}$
\STATE Initialize available microbatch queues $Q_{1} \leftarrow \{1, 2, \ldots, M\}$ and $Q_{j} \leftarrow \emptyset, \forall j \in \mathcal{J}/\{1\}$
\WHILE {$\exists j \in \mathcal{J}, Q_{j} \neq \emptyset$}
	\FOR {$j \in \mathcal{J}: Q_{j} \neq \emptyset$}
			\STATE \revise{Pop one microbatch $m$ out from the head of $Q_{j}$, and push $m$ to the end of $Q_{j+1}$ if $j < |\mathcal{J}|$}
			\STATE Add $(m,j)$ to the corresponding $U_{s}$ if block $j$ is a computation block
	\ENDFOR
\ENDWHILE
\STATE Pop the first $(1, 1)$ out of $U_{s_1}$, and set $e^f_{1,s_1}=0$
\WHILE {$\exists s \in \mathcal{S}$, a block of stage $s$ completes or a communication block which transmits data to stage $s$ finishes at time $t$}
    \IF {$s \in \mathcal{S}_{repl}$ and $U_{s} = \emptyset$}
		\STATE Start AllReduce operation, and set $e^A_{s}$ to $t$
	\ENDIF
	\IF {$U_{s} = \emptyset, \forall s \in \mathcal{S}$}
		\STATE \textbf{break}
	\ENDIF
	\IF {a computation block of $s$ finishes}
		\STATE Start successor communication block
	\ENDIF
	\IF {the first (microbatch index, block number) in $U_{s}$ is ready}
		\STATE Pop (microbatch index, block number) out of $U_{s}$
		\STATE Start the block and set the $e^f_{m,s}$ or $e^b_{m,s}$ to $t$
	\ENDIF
\ENDWHILE
\STATE Calculate the makespan: {\footnotesize $T_{PE}=\max\{\max_{m\in\{1, 2, \ldots, M\}} (e^b_{m, s_1} + \frac{\sum\limits_{l \in s_{1}}p^b_l}{|\mathcal{F}({s_1})|}), \max_{s\in\mathcal{S}_{repl}}(e^A_{s} + A_s)\}$} 
\STATE Return $T_{PE}$, $\textbf{e}$
\end{algorithmic}
\end{algorithm}
\setlength{\textfloatsep}{0pt}
}

We summarize our pipeline execution scheduling algorithm in Alg.~\ref{alg_schedule}.  
The following lemma gives an upper bound of the per-iteration training time achieved by this PE algorithm.

\begin{lemma}
\label{lemma_ss}
Per-iteration training time achieved by Alg.~\ref{alg_schedule}, $T_{PE}$, is no larger than {\small $(1 + \frac{4|\mathcal{S}|-4}{M})M\mathcal{C} + \max_{s\in\mathcal{S}_{repl}}\{A_s\}$}, where {\small $\mathcal{C} = \max\{\max_{n \in \{1, \ldots, |\mathcal{S}|\}} \frac{\sum_{l \in s_n}(p^f_l + p^b_l)}{|\mathcal{F}(s_n)|}, \allowbreak\max_{n\in \{1, \ldots, |\mathcal{S}|-1\} 
	}\{c^f_{s_{n}, s_{n+1}}+c^b_{s_{n+1}, s_n}\}\}$}, denoting the maximum time to process a microbatch on a single stage (including both forward and backward computation) or an inter-stage communication channel (including data transfer in both forward and backward propagation phases), without considering AllReduce operations.
\end{lemma}

\begin{proof}
{Given the execution order computed with lines 3-8 in Alg.~\ref{alg_schedule}, we consider a new \textit{cycle scheduling algorithm} whose per-iteration training time, denoted by $T_{CS}$, serves as an upper-bound of $T_{PE}$. 
In every cycle, we schedule to execute every computation/communication block $j \in \mathcal{J}$ for at most one microbatch if available.
Our cycle scheduler starts by entering the first cycle, and only schedules the execution of the first microbatch for the forward computation block of stage 1 as all other blocks are not available. After the execution of all available blocks in the current cycle, our cycle scheduler transits to the next cycle, and checks the availability of all blocks again. If a block has at least one microbatch that is available, we schedule the execution of one of the available microbatches for the block. The scheduler ends when every microbatch has been processed by all the blocks,
In addition, for each replicated stage, we execute the corresponding AllReduce operation immediately upon all microbatches have been processed by backward computation block of this stage. 
Consequently, the execution time of every cycle is at most 
{\small $\mathcal{C} = \max\{\max_{n \in \{1, \ldots, |\mathcal{S}|\}} \frac{\sum_{l \in s_n}(p^f_l + p^b_l)}{|\mathcal{F}(s_n)|}, \allowbreak\max_{n\in \{1, \ldots, |\mathcal{S}|-1\}}\{c^f_{s_{n}, s_{n+1}}+c^b_{s_{n+1}, s_n}\}\}$}, representing the maximum time to process a microbatch on a single stage (including both forward and backward computation) or an inter-stage communication channel (including data transfer in both forward and backward propagation phases), without considering AllReduce operations.

Fig.~\ref{fig_pipeline_cycle} shows an example of our cycle-based schedule of six microbatches in one training iteration. The DNN model is partitioned into two stages where the second stage is replicated over multiple GPUs. In the first cycle, we only process microbatch 1 on the forward computation block in stage 1. Then in the second cycle, two blocks have available microbatches to be executed: microbatches 2-6 on forward computation block in stage 1; and microbatch 1 on forward communication block in communication channel 1. Hence, we schedule the execution of the two blocks once, \ie~CF1 and F2 in the figure. It takes 10 cycles to execute all microbatches on all the blocks.
Every cycle starts upon all the blocks scheduled in the previous cycle finishes, and strictly schedules to execute at most one microbatch on each block.


\begin{figure*}[!t]
  \centering
  \includegraphics[width=0.8\textwidth]{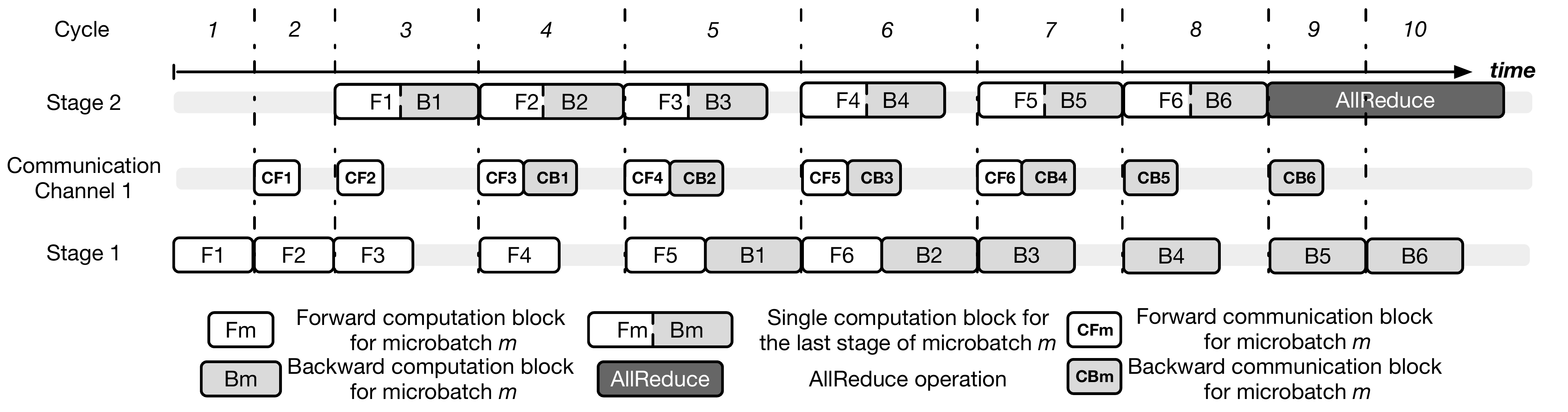}
  \caption{Cycle schedule: an example}
  \label{fig_pipeline_cycle}
\end{figure*}

Given that the DNN model is partitioned into $N$ stages with $M$ microbatches, there are two cases:
	
	$\triangleright$ \textbf{Case 1:} $M > 4|\mathcal{S}|-4$
	
	In this case, we first perform $4|\mathcal{S}|-4$ cycles before entering a cycle where every block $j \in \mathcal{J}$ has at least one available microbatch. Afterwards, we perform another $M-4|\mathcal{S}|+4$ cycle until all the microbatches have been processed by the first block. We further perform $4|\mathcal{S}|-4$ cycles to finish all microbatches.
	
	$\triangleright$ \textbf{Case 2:} $M \leq 4|\mathcal{S}|-4$
	
	Similarly to the first case, we perform in total $M + 4|\mathcal{S}| - 4$ cycles to execute all the microbatches.
	
In conclusion, we perform in total $M + 4|\mathcal{S}| - 4$ cycles in either cases, with the time for each cycle no greater than $\mathcal{C}$.
As a result, we have:

{\small
\begin{displaymath}
	T_{CS} \leq (M + 4|\mathcal{S}| - 4)\mathcal{C} + \max_{s\in\mathcal{S}_{repl}}\{A_s\}
\end{displaymath}}

The per-iteration training time achieved by the cycle scheduling algorithm, $T_{CS}$, is an upper bound of the per-iteration training time achieved by Alg.~\ref{alg_schedule}, $T_{PE}$. This is because Alg.~\ref{alg_schedule} schedules the execution of microbatches on the same block without interruption if available, while in the cycle scheduling algorithm, the execution of blocks waits until all blocks in the previous cycle have been executed. 
We have proven that

{\small
\begin{displaymath}
	T_{PE} \leq T_{CS} \leq (1 + \frac{4|\mathcal{S}|-4}{M})M\mathcal{C} + \max_{s\in\mathcal{S}_{repl}}\{A_s\}
\end{displaymath}}
}
\end{proof}

\vspace{-2mm}
\subsection{
DNN Partition and Device Mapping Algorithm} 
\label{sec::part}
\vspace{-1mm}

Lemma~\ref{lemma_ss} shows that the per-iteration training time 
is positively related to the number of stage $|\mathcal{S}|$ that DNN $\mathcal{D}$ is partitioned into, the maximum time $\mathcal{C}$ to process a microbatch on a single stage or communication channel, and the maximum AllReduce operation time among replicated stages. 
The number of stage partitions, $|\mathcal{S}|$, is at most the same as the number of GPUs, $V$. We next design a model partition and device mapping algorithm aiming at minimizing the maximum time, $\mathcal{W}(|\mathcal{S}|)$, to process all microbatches on a single stage (including AllReduce operations) or 
a 
communication channel, given the number of stages $|\mathcal{S}|$. 
The purpose is to minimize the upper bound of per-iteration training time in Lemma \ref{lemma_ss}, as $\mathcal{W}(|\mathcal{S}|)$ is related to both $\mathcal{C}$ and $\max_{s\in\mathcal{S}_{repl}}\{A_s\}$: 

\vspace{-4mm}
{\small
\begin{eqnarray}
	\mathcal{W}(|\mathcal{S}|) = \max\{&\max_{s\notin \mathcal{S}_{repl}}\{M\sum\limits_{l\in s}(p^f_l + p^b_l)\},\nonumber\\
	&\max_{s\in \mathcal{S}_{repl}}\{M\frac{\sum\limits_{l\in s}(p^f_l + p^b_l)}{|\mathcal{F}(s)|} + A_s\},\nonumber\\
	&\max_{n\in \{1,\ldots, |\mathcal{S}|-1\}} 
				 \{M(c^f_{s_n, s_{n+1}}+c^b_{s_{n+1}, s_n})\}\}\nonumber
\end{eqnarray}}
\vspace{-6mm}

Our DNN partition and device mapping problem, 
without considering stage replication, can be reduced to the NP-complete problem of pipeline partition over a heterogeneous communication platform~\cite{benoit2008mapping}, which partitions a workflow among a cluster of devices with heterogeneous connectivity to maximize the pipeline throughput. We design an efficient \textit{balanced pipeline partition and device mapping algorithm} (BPPM) 
to derive a near-optimal solution, which includes two components: 1) a device ordering module that calculates a linear ordering of all GPUs; 
and 2) an 
algorithm that partitions the DNN model onto GPUs respecting the device order.

\textit{\textbf{1) Recursive device ordering (RDO)}}: 
We decide a linear ordering of GPUs in $\mathcal{V}$: 
$\{v_1, v_2, \ldots, v_{V}\}$. 
We will map stages (stage replicas) to devices according to this 
ordering and stage dependencies, \ie, (replicas of) the first stage mapped to the device(s) at the head of the device ordering, and then the next stage to latter device(s), etc.
We target a linear ordering with maximal bandwidth between consecutive GPUs, 
such that the bandwidth between stages and between replicas of the same stage is maximized. 
A recursive min-cut algorithm is designed to find the ordering 
in polynomial time, as given in Alg.~\ref{alg_do}. $rank_l$ ($rank_h$) represents the lowest (highest) rank of devices in the current subgraph in the ordering, initialized to $1$ and $V$, respectively (in the complete Alg.~\ref{alg_spp} that invokes RDO). In each recursive step, we find a min-cut within the current input graph (leveraging an efficient min-cut algorithm in~\cite{stoer1997simple}), 
 partition the graph into two subgraphs accordingly, and call RDO again to order devices in the two subgraphs, respectively. When the input graph contains only one GPU, we assign $rank_l$ (which equals $rank_h$) to it. We order GPUs according to their computed ranks in ascending order.

{\small
\begin{algorithm}[!t]

\caption{Recursive Device Ordering 
- \textbf{\textit{RDO}}}
\label{alg_do}

\renewcommand{\algorithmicrequire}{\textbf{Input:}}
\renewcommand{\algorithmicensure}{\textbf{Output:}}

\begin{algorithmic}[1]
\REQUIRE $G(\mathcal{V}, \mathcal{E}), rank_l, rank_h$ 

\IF {$|\mathcal{V}| == 1$}
	\STATE Set $rank(v) \leftarrow rank_l$
\ELSE
	\STATE $G^1(\mathcal{V}^1, \mathcal{E}^1), G^2(\mathcal{V}^2, \mathcal{E}^2) = \text{min-cut}(G)$
	\STATE \textbf{\textit{RDO}}$(G^1, rank_l, rank_l + |\mathcal{V}^1| - 1)$
	\STATE \textbf{\textit{RDO}}$(G^2, rank_l + |\mathcal{V}^1|, rank_h)$
\ENDIF
\end{algorithmic}
\end{algorithm}
}

The link(s) within the min-cut will be used for inter-stage communication between the two consecutive stages mapped respectively onto two GPUs in the two subgraphs, or AllReduce operation within one stage whose replications are mapped into both subgraphs. Since all GPUs will be used in pipeline training, 
at least one link in each min-cut needs to be used for communication 
(as otherwise the training topology will not be a connected graph).
By dividing the device graphs in this way, link(s) in each min-cut will be used only between two stages or among replicas of one replicated stage, but not between two pairs of consecutive stages or replicas of two replicated stages. 
Hence, links with small bandwidth are minimally exploited while large-bandwidth links are maximally used for inter-stage or AllReduce 
communication, minimizing the maximum communication time on a single communication channel. 

\sloppy \textit{\textbf{2) Pipeline partition, replication and mapping (PRM)}}: Following the device ordering $\{v_1, v_2, \ldots, v_{V}\}$, we leverage dynamic programming to partition and map $\mathcal{D}$ onto the GPUs that minimizes the maximum execution time $\mathcal{W}(|\mathcal{S}|)$ on a single stage or communication channel.

Let $W(l, \xi, r, i)$ denote the optimal (aka minimum) maximum execution time on a single GPU or communication channel, when we partition the first $l$ consecutive layers in $\mathcal{D}$ into $\xi$ stages with the last stage 
replicated into $r$ replica(s) ($r \ge 1$), and place the stages over GPUs $v_1$ to $v_i$. 
We have $\mathcal{W}(|\mathcal{S}|) = \min_{1\le r\le V}W(L, |\mathcal{S}|, r, V)$. $W(l, \xi, r, i)$ can be recursively computed as follows:

\vspace{-5mm}
{\small
\begin{eqnarray*}
	\label{eqn_dp}
	&W(l, \xi, r, i) = \min\limits_{1\le l'\le l-1, 1\le r'\le i-r}\max\{W(l', \xi-1, r', i-r), \\
	&M\frac{d^f_{l', l'+1} + d^b_{l'+1, l'}}{r'rb_{r'r}}, M\frac{\sum\limits_{o = l'+1}^{l}(p^f_o + p^b_o)}{r} + A_{l'+1\rightarrow l}(v_{i-r+1} \rightarrow v_{i})\}
\end{eqnarray*}}
\vspace{-6mm}

\noindent The first term inside $\max$ is the maximal execution time on a single GPU/communication channel by optimal partition of layers $1$ to $l'$ into $\xi-1$ stages (with the last stage replicated on $r'$ GPUs) and mapping them on GPUs $v_1$ to $v_{i-r}$. The second term computes the total communication time on the communication channel between layers $l'$ and $l'+1$, where 
	$b_{r'r} = \min_{v'\in \{v_{i-r-r'+1}, \ldots, v_{i-r}\}, v\in \{v_{i-r+1}, \ldots, v_{i}\}}b_{v', v}$ is the minimal link bandwidth between $r'$ replicas of layer $l'$ and $r$ replicas of layer $l'+1$.
The third term is the training time on the last stage, including processing time of all microbatches over layers $l'+1$ to $l$ replicated on $r$ GPUs 
and time taken by the corresponding AllReduce operation. Here $A_{l'+1\rightarrow l}(v_{i-r+1} \rightarrow v_{i})$ denotes the time for AllReduce operation of layers $l'+1$ to $l$ replicated over GPUs $v_{i-r+1}$ to $v_i$. 

To compute $W(l, \xi, r, i)$, we solve the subproblem of optimally partitioning the first $l'$ layers into $\xi-1$ stages on GPUs $v_1$ to $v_{i-r}$, while replicating the stage containing layers $l'+1$ to $l$ over GPUs $v_{i-r+1}$ to $v_i$. We consider all possible choices of layer $l'$ and various replication strategies of the stage containing $l'$, and decide $W(l, \xi, r, i)$ as the minimal time computed among them.
The detailed dynamic programming PRM algorithm is given in Appendix~\ref{appendix_alg_dp}.
 The following lemma shows that the best stage partition and device mapping that our algorithms identify when the given number of stages varies, achieves a maximum per-stage/communication channel execution time close to optimum.

\vspace{-2mm}
\begin{lemma}
\label{lemma_bmmp}
	Let $\mathcal{W}_{PRM}=\min_{|\mathcal{S}| \in \{1,\ldots, V\}}\mathcal{W}(|\mathcal{S}|)$. 
	$\mathcal{W}_{PRM}$  achieved by RDO and PRM is no larger than $(1 + \Phi)$ times $\mathcal{W}^*$, the optimal (aka minimum) maximum execution time on a single stage or communication channel, 
	where $\Phi=\frac{\max\{p_{max}b_{max}, d_{max}\}}{\Gamma}(\frac{1}{b_{min}} - \frac{1}{b_{max}})$, {\small $d_{max}=\max_{1\le l\le L-1} (d^f_{l,l+1} + d^b_{l+1,l})$, $p_{max}=\max_{1\le l \le L} (p^f_{l} + p^b_{l})$}, and {\small $\Gamma = {\sum\limits_{1\le l \le L}(p^f_l + p^b_l)}/{V}$}. 
\end{lemma}

\begin{proof}
{
	Consider a new multi-GPU system graph ${\mathcal{G}_{max}}(\mathcal{V}, {\mathcal{E}_{max}})$, where the vertices in ${\mathcal{G}_{max}}$ are the same as in $\mathcal{G}$ while the bandwidth of every edge in ${\mathcal{E}_{max}}$ equals the maximum bandwidth $b_{max}$.
	
We consider the pipeline partition and mapping solution of $\mathcal{D}$ with ${\mathcal{G}}_{max}$ instead of the original $\mathcal{G}$, that minimizes the maximum total execution/communication time on a single stage or communication channel (regardless of the number of stages). We denote the solution as ${\mathcal{S}}_{max}$ and ${\mathcal{F}}_{max}$, and the optimal maximum total execution/communication time on a single stage or communication channel as ${\mathcal{W}}^*_{max}$. There are two cases of the solution:
	
	\noindent$\triangleright$ \textbf{Case 1:} There is at least one stage replicated in the solution.
	
	Without loss of generality, we consider one stage $s$ comprised of layer $l_1$ to layer $l_2$ is replicated over $k$ GPUs. Based on Eqn.~(\ref{eqn_allreduce}), we derive the total execution time $\Psi_s$ of $s$ on the cluster of $k$ GPUs:
	
	{\small
	\begin{equation}
		\label{eqn_proof_bmmp_1}
		\Psi_s = M\frac{\sum\limits_{l\in s}(p^f_l + p^b_l)}{k} + \frac{2(k-1)\sum\limits_{l\in s}\alpha_l}{k b_{max}}
	\end{equation}}
	
	On the other hand, we can construct a partition and mapping solution of layers in stage $s$ without replication. The construction is as follows. We firstly allocate the first $x$ layers in $s$ onto the first GPU until {\small$M\sum\limits_{l=l_1}^{l_{x-1}}(p^f_l + p^b_l) \leq M(\frac{\sum\limits_{l\in s}(p^f_l + p^b_l)}{k})$ while $M\sum\limits_{l=l_1}^{l_x}(p^f_l + p^b_l) > M(\frac{\sum\limits_{l\in s}(p^f_l + p^b_l)}{k})$}. Then, we consider the layer allocation on next GPU following the same procedure. We perform the above allocation iteratively until all the layers have been mapped to GPUs. It is clear that our construction can map all $l_1$ to $l_2$ layers onto the $k$ GPUs. Otherwise, if there are layers left, the sum of the execution time of all the allocated layers should be larger than {\small$M{\sum\limits_{l\in s}(p^f_l + p^b_l)}$}, yielding a contradiction. 
	In addition, our construction ensures that the maximum execution time on any GPU is no large than $M(\frac{\sum\limits_{l\in s}(p^f_l + p^b_l)}{k} + p'_{max})$, where $p'_{max} = \max\limits_{l_1\leq l\leq l_2}(p^f_l + p^b_l)$. 
	We denote the maximum total execution/communication time on a single stage or communication channel achieved with the construction as $\hat{\Psi}_s$. We have
	{\small
	\begin{multline}
		\label{eqn_proof_bmmp_2}
		\Psi_s \leq   \hat{\Psi}_s \\
				 = M\max\{(\frac{\sum\limits_{l = l_1}^{l_2}(p^f_l + p^b_l)}{k} + p'_{max}), \frac{\max_{l_1\leq l \leq {l_2+1}}\{d^f_{l-1,l} + d^b_{l-1,l}\}}{b_{max}}\}
	\end{multline}}
	
	Combining Eqn.~(\ref{eqn_proof_bmmp_1}) and Eqn.~(\ref{eqn_proof_bmmp_2}) and noting that $p'_{max} \leq p_{max}$, we derive that for any replicated stage $s$, the data volume transmitted in its AllReduce operation, $d_s^{AR}$, is:
	{\small
	\begin{align}
		d_s^{AR} &= \frac{2(k-1)\sum\limits_{l\in s}\alpha_l}{k} \nonumber\\
			&\leq M\max\{p'_{max}b_{max}, {\max_{l\in s}\{d^f_{l-1,l} + d^b_{l-1,l}\}} - \frac{\sum\limits_{l\in s}(p^f_l + p^b_l)}{k}b_{max}\}\nonumber\\
			& \leq M\max\{p_{max}b_{max}, d_{max}\}\label{eqn_proof_bmmp_3}
	\end{align}}
	
	Now, let us substitute ${\mathcal{E}}_{max}$ in ${\mathcal{G}}_{max}$ by ${\mathcal{E}}$ while keeping the solution unchanged and respecting the device ordering $\{v_1, v_2, \ldots, v_N\}$. The maximum execution/communication time on a single stage or communication channel in this case is denoted as $\mathcal{{W}}'$. As we only change the bandwidths in $\mathcal{G}_{max}$, we have:
	{\small
	\begin{align}
		\mathcal{{W}}' &\leq {\mathcal{W}}^*_{max} + \max\{Md_{max}, \max_{s\in \mathcal{S}_{repl}}\{d_s^{AR}\}\}(\frac{1}{b_{min}} - \frac{1}{b_{max}}) \nonumber\\
			& \leq {\mathcal{W}}^* + M\max\{p_{max}b_{max}, d_{max}\}(\frac{1}{b_{min}} - \frac{1}{b_{max}})\label{eqn_proof_bmmp_4}
	\end{align}}
	where the first inequality is because we consider the maximal possible increment in execution time due to the decrement in bandwidth, and the second inequality is due to Eqn.~(\ref{eqn_proof_bmmp_3}) and ${\mathcal{W}}^*_{max} \leq \mathcal{W}^*$.

	In addition, we have {\small$\mathcal{W}_{PPM} \leq \mathcal{{W}}'$}.
	Noting that $M\Gamma$, denoting the case of evenly distributed all workload across $N$ GPUs, is a lower bound of $\mathcal{W}^*$, we have:
	
	{\small
	\begin{align}
		\mathcal{W}_{PRM} &\leq \mathcal{{W}}' \leq {\mathcal{W}}^* + M\max\{p_{max}b_{max}, d_{max}\}(\frac{1}{b_{min}} - \frac{1}{b_{max}})\nonumber\\
			& = \mathcal{W}^* + M\frac{\max\{p_{max}b_{max}, d_{max}\}}{\mathcal{W}^*}(\frac{1}{b_{min}} - \frac{1}{b_{max}})\mathcal{W}^* \nonumber\\
			&\leq (1 + \frac{\max\{p_{max}b_{max}, d_{max}\}}{\Gamma}(\frac{1}{b_{min}} - \frac{1}{b_{max}}))\mathcal{W}^* \label{eqn_proof_bmmp_5}
	\end{align}}
	
	\noindent$\triangleright$ \textbf{Case 2:} There is no stage replicated in the solution. Following similar steps as in case 1, we have:
	{\small
	\begin{align}
		\mathcal{W}_{PPM} &\leq \mathcal{{W}}' \leq {\mathcal{W}}^*_{max} + Md_{max}(\frac{1}{b_{min}} - \frac{1}{b_{max}})\nonumber\\
			 &\leq (1 + \frac{d_{max}}{\Gamma}(\frac{1}{b_{min}} - \frac{1}{b_{max}}))\mathcal{W}^*\label{eqn_proof_bmmp_6}
	\end{align}}
	where the second inequality is because we only consider inter-stage communication in this case, and the last inequality is due to $M\Gamma \leq \mathcal{W}^*$.
	
	We combine Eqn.~(\ref{eqn_proof_bmmp_5}) and Eqn.~(\ref{eqn_proof_bmmp_6}) to derive:
	
	{\small
	\begin{equation}
		\mathcal{W}_{PRM} \leq (1 + \frac{\max\{p_{max}b_{max}, d_{max}\}}{\Gamma}(\frac{1}{b_{min}} - \frac{1}{b_{max}}))\mathcal{W}^*
	\end{equation}
	}
	
}
\end{proof}

\vspace{-4mm}
\subsection{Complete Synchronous Pipeline Planning Algorithm}
\label{sec::complete_alg}

\begin{algorithm}[!t]
\caption{Synchronous Pipeline Planning - \textbf{\textit{SPP}}}
\label{alg_spp}

\renewcommand{\algorithmicrequire}{\textbf{Input:}}
\renewcommand{\algorithmicensure}{\textbf{Output:}}

\begin{algorithmic}[1]

\REQUIRE $G(\mathcal{V}, \mathcal{E}), \mathcal{D}$
\ENSURE $\mathcal{\bar{S}}, \mathcal{\bar{F}}, \bar{\textbf{e}}, T_{SPP}$ 
\STATE \textbf{\textit{RDO}}($\mathcal{G}(\mathcal{V}, \mathcal{E}), 1, V$)
\STATE Obtain device ordering $\{v_1, \ldots, v_V\}$ according to $rank(v), \forall v \in \mathcal{V}$
\STATE $T_{SPP} \leftarrow \textbf{INF} $
\FOR {
$\xi \in \{1, 2, \ldots, V\}$}
	\FOR {$r \in \{1, 2, \ldots, V\}$}
		\STATE $W(L, \xi, r, V), \mathcal{S}_{r}, \mathcal{F}_{r} \leftarrow$ \textbf{\textit{PRM}}($G(\mathcal{V}, \mathcal{E}), \{v_1, \ldots, v_V\}, \mathcal{D}, L, V, \xi, r$)
	\ENDFOR
	\STATE Set $\mathcal{S}$ and $\mathcal{F}$ to $\mathcal{S}_{r}$ and $\mathcal{F}_{r}$ that achieve the minimum $W(L, \xi, r, V)$
	\STATE $T_{PE}, \textbf{e} \leftarrow$\textbf{\textit{PE}}($\mathcal{G}, \mathcal{S}, \mathcal{F}$)
		\STATE $T_{SPP} \leftarrow T_{PE}, \mathcal{\bar{S}} \leftarrow \mathcal{S}, \mathcal{\bar{F}} \leftarrow \mathcal{F}, \bar{\textbf{e}} \leftarrow \textbf{e}$ if $T_{PE} < T_{SPP}$
\ENDFOR
\STATE Return $\mathcal{\bar{S}}, \mathcal{\bar{F}}, \bar{\textbf{e}}, T_{SPP}$
\end{algorithmic}
\end{algorithm}

\vspace{-1mm}
Our complete synchronous pipeline planning ({\textit{SPP}}) algorithm is given in Alg.~\ref{alg_spp}, which produces model partition $\mathcal{\bar{S}}$, device mapping $\mathcal{\bar{F}}$ and execution schedule $\bar{\textbf{e}}$. We first leverage RDO in Alg.~\ref{alg_do} to obtain a linear ordering of all GPUs (lines 1-2). We next vary the number of stages 
from $1$ to $V$ (line 4): 
given a stage number to partition the model into, we vary $r$ and call PRM to compute the best stage partition and device mapping (lines 5-8) that achieve $\mathcal{W}(|\mathcal{S}|)$; we then invoke PE in Alg.~\ref{alg_schedule} to compute execution schedule of microbatches over these partitions on the respective devices (line 9). 
We identify the best stage partition number as the one minimizing the makespan of a training iteration (lines 10-12) together with the corresponding $\mathcal{\bar{S}}$, $\mathcal{\bar{F}}$ and $\bar{\textbf{e}}$. 

\vspace{-2mm}
\begin{theorem}
	The makespan of a training iteration achieved by {\textit{SPP}}, $T_{SPP}$, is less than $(2 + \frac{4V-4}{M})(1 + \Phi)$ times the optimal makespan, $T^*$.
	\label{th_approx_ratio}
\end{theorem}
\vspace{-2mm}
\begin{theorem}
	Our complete synchronous pipeline planning Alg.~\ref{alg_spp} runs in polynomial time.
\end{theorem}
\vspace{-4mm}

\section{Performance Evaluation}
\label{sec::eval}
\vspace{-2mm}

We evaluate {\textit{SPP}} with both testbed experiments and simulation studies. 
\vspace{-2mm}
\subsection{Testbed experiments}

\noindent\textbf{Implementation.} 
We implement {\textit{SPP}} using C++ and Python on Tensorflow 1.14.1~\cite{abadi2016tensorflow}. We use Tensorflow profiler to collect runtime data of each layer of each DNN model (\eg~forward/backward computation time, parameter size and activation size) over 20 training iterations. 
We 
assign a priority to each stage or AllReduce operation (implemented using NCCL collective AllReduce~\cite{NCCL}) based on our computed execution order, such that they can be scheduled by TensorFlow execution engine accordingly.  

\noindent\textbf{Testbed.} 
We evaluate {\textit{SPP}} in two testbed environments: (1) One consists of 4 GPU servers, inter-connected by a Dell Z9100-ON switch, with 50Gbps peak bandwidth between any two servers. Each server has one 8-core Intel E5-1660 CPU, two GTX 1080Ti GPUs and one 50GbE NIC. (2) The other 
is a single server equipped with 4 Tesla V100 GPUs, two 10-core Intel Xeon 
E5-2630 v4 CPUs and a 100GbE NIC. GPUs in the sever are connected with 128Gbps PCIe bus.

\noindent\textbf{DNN models.} We train 7 representative 
DNN models: three image classification models 
on the ImageNet dataset~\cite{deng2009imagenet} and four NLP models 
on SQuAD2.0 dataset~\cite{rajpurkar2018know} (Table~\ref{table_model}). 
The number of microbatches and microbatch size for training each model are set as the maximum number$\times$size (aka mini-batch size) without causing OOM (out of memory) for most baselines. 
The large batch sizes we use are consistent with common practice
~\cite{fan2021dapple}.
To run {\textit{SPP}}, we modified ResNet152 by ignoring shortcut connections and Inception-V3 by aggregating parallel branches (branches with the same start point and end point) as one layer. We apply \textit{SPP} to the modified models to decide the 
strategies, 
and then train the original models (without the modifications) using the obtained strategies. 

{\small
\begin{table}[!t]
\caption{Benchmark DNN models}
\begin{center}
\fontsize{8}{9}\selectfont
\begin{tabular}{|l|c|c|c|}
\hline
Model	     & \begin{tabular}[c]{@{}c@{}}\# of \\ parameters\end{tabular} & \begin{tabular}[c]{@{}c@{}}\# of \\ microbatches\end{tabular} & \begin{tabular}[c]{@{}c@{}}microbatch size\\ (\# of samples)\end{tabular} \\ \hline
VGG19~\cite{simonyan2014very}       & 144M             & 8                                                              & 32                                                         \\ \hline
ResNet152~\cite{he2016deep}    & 60M              & 4                                                              & 4                                                          \\ \hline
Inception-V3~\cite{szegedy2016rethinking} & 24M              & 8                                                              & 32                                                         \\ \hline
Transformer~\cite{vaswani2017attention}  & 55M              & 8                                                              & 32                                                         \\ \hline
BERT-large~\cite{devlin2018bert}   & 340M             & 4                                                              & 4                                                          \\ \hline
XLNet-large~\cite{yang2019xlnet}  & 550M             & 4                                                              & 4                                                          \\ \hline
BERT-48~\cite{devlin2018bert}   	 & 640M             & 4                                                              & \begin{tabular}[c]{@{}l@{}}4 - 1080Ti$\times$8\\ 2 - V100$\times$4  \end{tabular}                                                          \\ \hline
\end{tabular}
\label{table_model}
\end{center}
\vspace{-2mm}
\end{table}}

\noindent\textbf{Baselines.} {\em SPP} is compared with 4 state-of-the-art schemes: 
(i) Data Parallelism (DP), with 
each GPU training the complete model 
with $\frac{\mbox{mini-batch size}}{\mbox{\# of GPUs}}$ amount of data; 
(ii) GPipe~\cite{huang2019gpipe}; 
(iii) PipeDream~\cite{harlap2018pipedream}; 
(iv) HetPipe~\cite{Park2020hetpipe} 
(see Sec.~\ref{ppparallelism} for details of the latter three).
Unless stated otherwise, we enforce a synchronization barrier at the end of each training iteration in PipeDream and HetPipe, 
removing the negative impact of asynchronous training on model convergence.

\begin{figure*}[!t]
\begin{minipage}{0.25\textwidth}
\includegraphics[width=1\hsize]{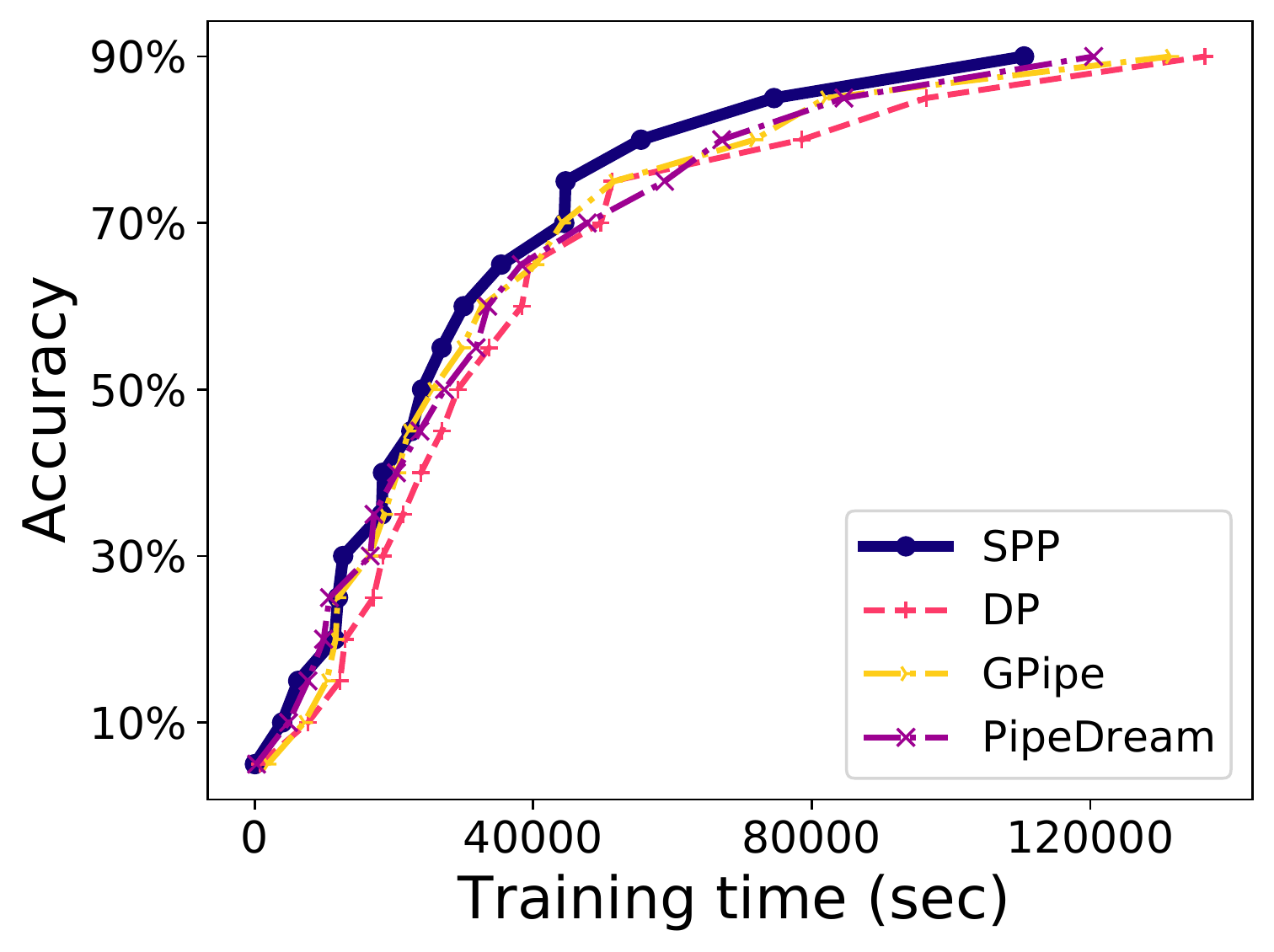}
  \captionof{figure}{VGG19 training progress: \revise{\textit{SPP}} vs. baselines}
  \label{vgg_converge}
\end{minipage}
    \hfill
\begin{minipage}{0.73\textwidth}
{\fontsize{8}{9}\selectfont
\begin{tabular}{|c|c|c|c|c|c|c|}
\hline
Model                         & Testbed           & \textit{SPP}   & \begin{tabular}[c]{@{}c@{}}DP\\ (Speed-up)  \end{tabular} & \begin{tabular}[c]{@{}c@{}}GPipe\\ (Speed-up)  \end{tabular}& \begin{tabular}[c]{@{}c@{}}PipeDream\\ (Speed-up)  \end{tabular} & \begin{tabular}[c]{@{}c@{}}HetPipe\\ (Speed-up)  \end{tabular} \\ \hline
\multirow{2}{*}{VGG19}        & 1080Ti$\times$8 & 1.799 & 2.882 (60.2\%)  & 2.120 (17.9\%)  & 1.949 (8.3\%)       & 2.696 (49.9\%)    \\ \cline{2-7} 
                              & V100$\times$4   & 0.983 & 1.245 (26.7\%)  & 1.004 (2.1\%)   & 1.024 (4.2\%)       & -                 \\ \hline
\multirow{2}{*}{ResNet152}    & 1080Ti$\times$8 & 0.732 & 0.896 (22.4\%)  & 1.214 (65.8\%)  & OOM (-)             & 0.843 (15.2\%)    \\ \cline{2-7} 
                              & V100$\times$4   & 0.832 & 1.209 (45.3\%)  & 1.041 (25.1\%)  & 0.873 (4.9\%)       & -                 \\ \hline
\multirow{2}{*}{Inception-V3} & 1080Ti$\times$8 & 0.303 & 0.420 (38.6\%)  & 0.551 (81.8\%)  & 0.656 (116.5\%)     & 0.408 (34.7\%)    \\ \cline{2-7} 
                              & V100$\times$4   & 0.357 & 0.663 (85.7\%)  & 0.587 (64.4\%)  & 0.919 (157.4\%)     & -                 \\ \hline
\multirow{2}{*}{Transformer}  & 1080Ti$\times$8 & 0.640 & 0.944 (47.5\%)  & 1.234 (92.8\%)  & 1.118 (74.7\%)      & 0.766 (19.7\%)     \\ \cline{2-7} 
                              & V100$\times$4   & 1.065 & 2.533 (137.8\%) & 1.487 (39.6\%)  & 1.830 (71.8\%)      & -                 \\ \hline
\multirow{2}{*}{BERT-large}   & 1080Ti$\times$8 & 0.409 & 0.524 (28.1\%)  & 0.472 (15.4\%)  & 0.421 (2.9\%)      & 0.525 (28.4\%)     \\ \cline{2-7} 
                              & V100$\times$4   & 0.952 & 2.269 (138.3\%) & 1.665 (74.9\%)  & 1.084 (13.9\%)      & -                 \\ \hline
\multirow{2}{*}{XLNet-large}  & 1080Ti$\times$8 & 1.299 & 1.388 (6.7\%)   & 1.696 (30.6\%)  & 1.384 (6.5\%)       & 1.628 (25.3\%)    \\ \cline{2-7} 
                              & V100$\times$4   & 1.437 & 1.842 (28.3\%)  & 1.720 (19.7\%)  & 1.690 (17.6\%)      & -                 \\ \hline
\multirow{2}{*}{BERT-48}   & 1080Ti$\times$8 & 0.762 & OOM (-)  & 1.885 (147.4\%)  & 1.266 (66.1\%)      &  1.377 (80.7\%)    \\ \cline{2-7} 
                              & V100$\times$4   & 0.855 & 1.656 (93.7\%) & 1.199 (40.2\%)  & 1.160 (35.7\%)      & -                 \\ \hline
\end{tabular}}

\captionof{table}{Per-iteration training time (in seconds) of different DNN models}
\label{table_iteration_time}
\end{minipage}
\vspace{-7mm}
\end{figure*}

\noindent\textbf{{Per-iteration training speed-up.}} We compare {\textit{SPP}} with all baselines in terms of per-iteration training time in the first testbed environment (1080Ti$\times$8). 
In the other environment (V100$\times$4), we 
omit HetPipe as all four GPUs are on the same server, reducing HetPipe to PipeDream solutions. 
In Table~\ref{table_iteration_time}, the speed-up is computed by $\frac{\text{Baseline time} - \text{SPP time}}{\text{SPP time}}$. 
{\textit{SPP}} outperforms the baselines in all cases. 
While both DP and HetPipe require an AllReduce operation to synchronize gradients, 
{\textit{SPP}} incurs less parameter synchronization traffic and maximally overlaps communication with computation within each training iteration. 
As a result, {\textit{SPP}} outperforms them 
by more than 20\% in most cases. 
The large speed-up for {\textit{SPP}} over baselines on Inception-V3 demonstrates that our design handles a model with non-uniform layer computation time well.
As VGG-19 has a small number of layers that can be easily optimally partitioned, we observe minor gain comparing {\em SPP} with PipeDream and GPipe. For BERT-large on 1080Ti$\times$8 testbed, PipeDream partitions the model into uniform stages, achieving similar good performance as {\em SPP}. 

{
\noindent\textbf{End-to-end training performance.} We next compare training convergence among \textit{SPP}, DP, GPipe and PipeDream. 
For PipeDream, we use its original asynchronous pipeline design. 
Fig.~\ref{vgg_converge} shows the training progress of VGG19 on the V100$\times$4 testbed to achieve a target 90\% top-5 accuracy~\cite{deng2009imagenet}. {\textit{SPP}} achieves the target accuracy using the least training time. Despite only marginal speed-up in per-iteration training time as compared to PipeDream (using its synchronous pipeline mode that we implemented), here {\textit{SPP}} outperforms PipeDream (with original asynchronous pipeline design) by 9.05\% in terms of the end-to-end training time. 
This is because 
PipeDream's asynchronous pipeline training slows down the model convergence progress, \revise{as training microbatches on outdated versions of model parameters~\cite{ho2013more}.}
}

\begin{figure}[!th]
\begin{minipage}[t]{0.45\columnwidth}
\includegraphics[width=\textwidth]{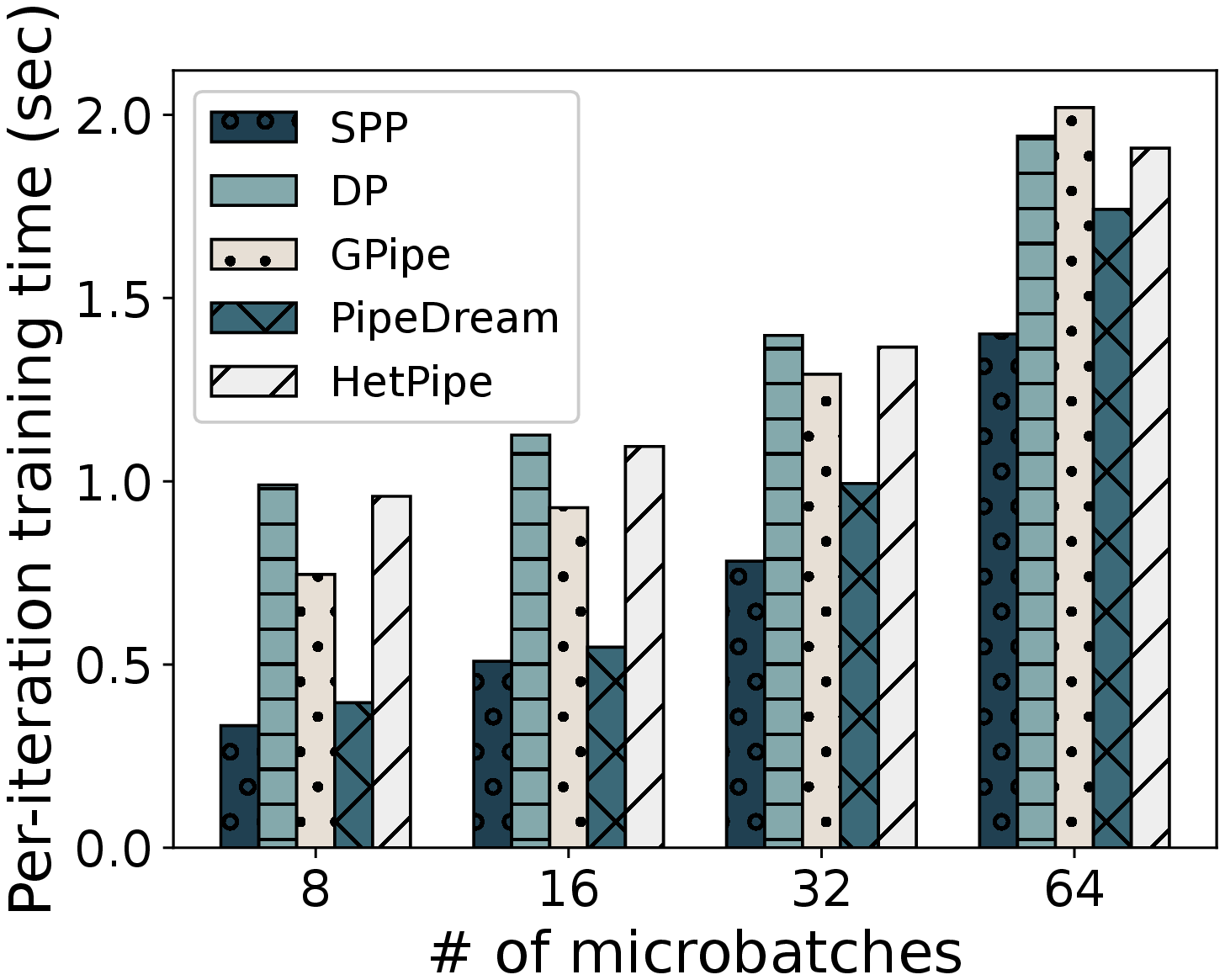}
\caption{Training time: different \# of microbatches}
\label{diff_batch}
\end{minipage}
\begin{minipage}[t]{0.45\columnwidth}
\includegraphics[width=\textwidth]{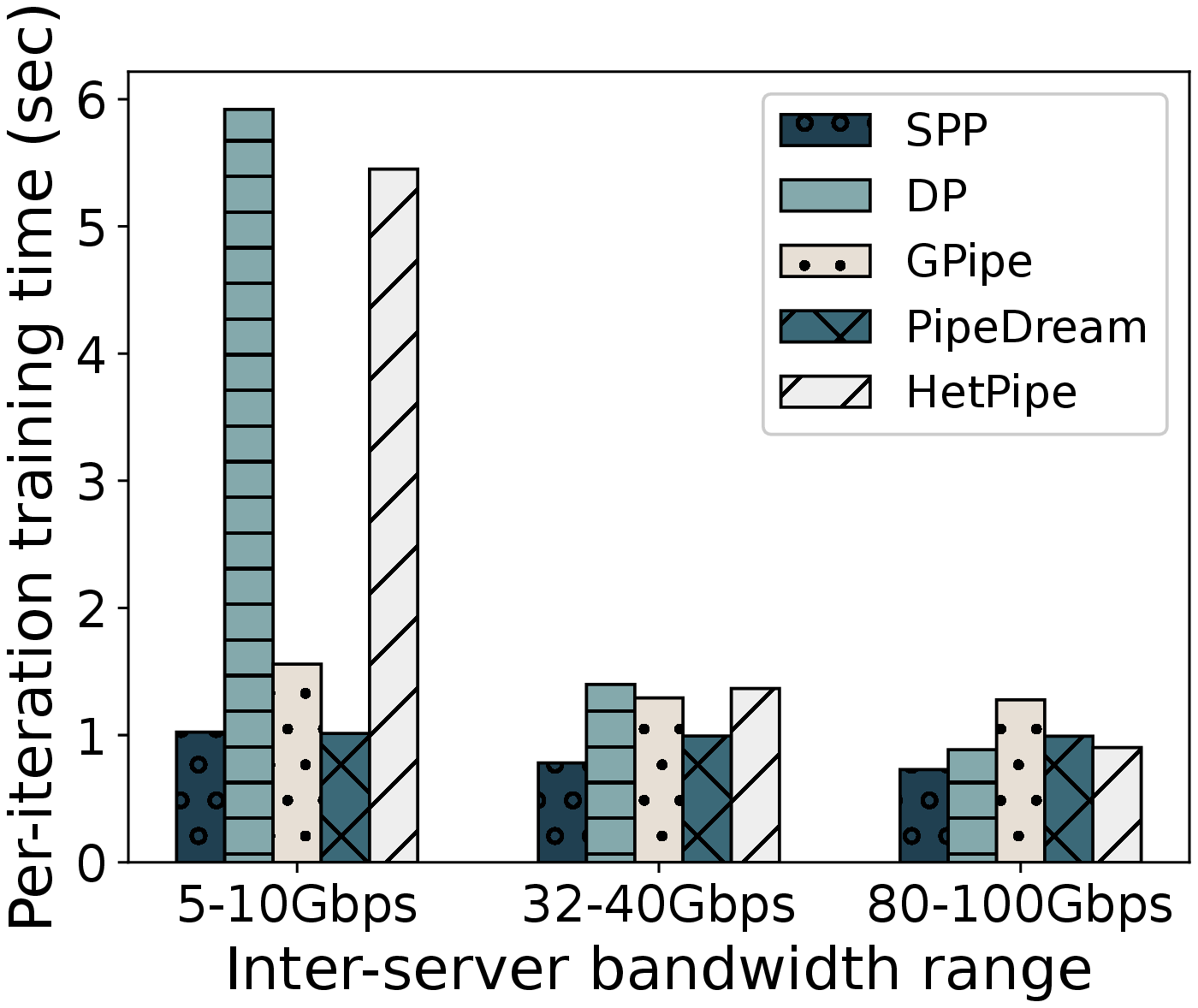}
\caption{Training time: different inter-server bandwidth levels}
\label{diff_band}
\end{minipage}
\begin{minipage}[t]{0.45\columnwidth}
\includegraphics[width=\textwidth]{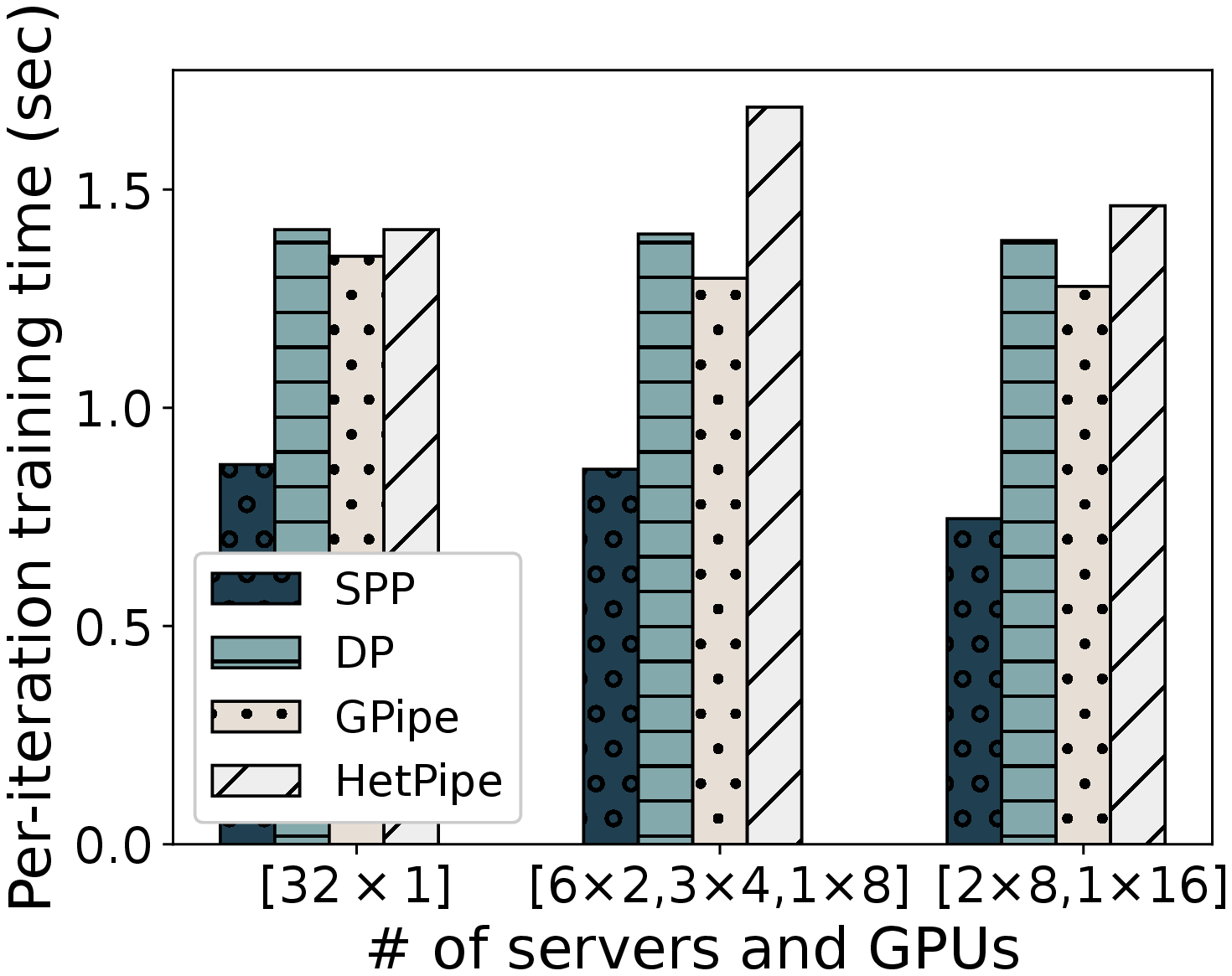}
\caption{Training time: different inter-GPU connectivity}
\label{diff_inter_connectivity}
\end{minipage}
\centering
\begin{minipage}[t]{0.45\columnwidth}
\includegraphics[width=\textwidth]{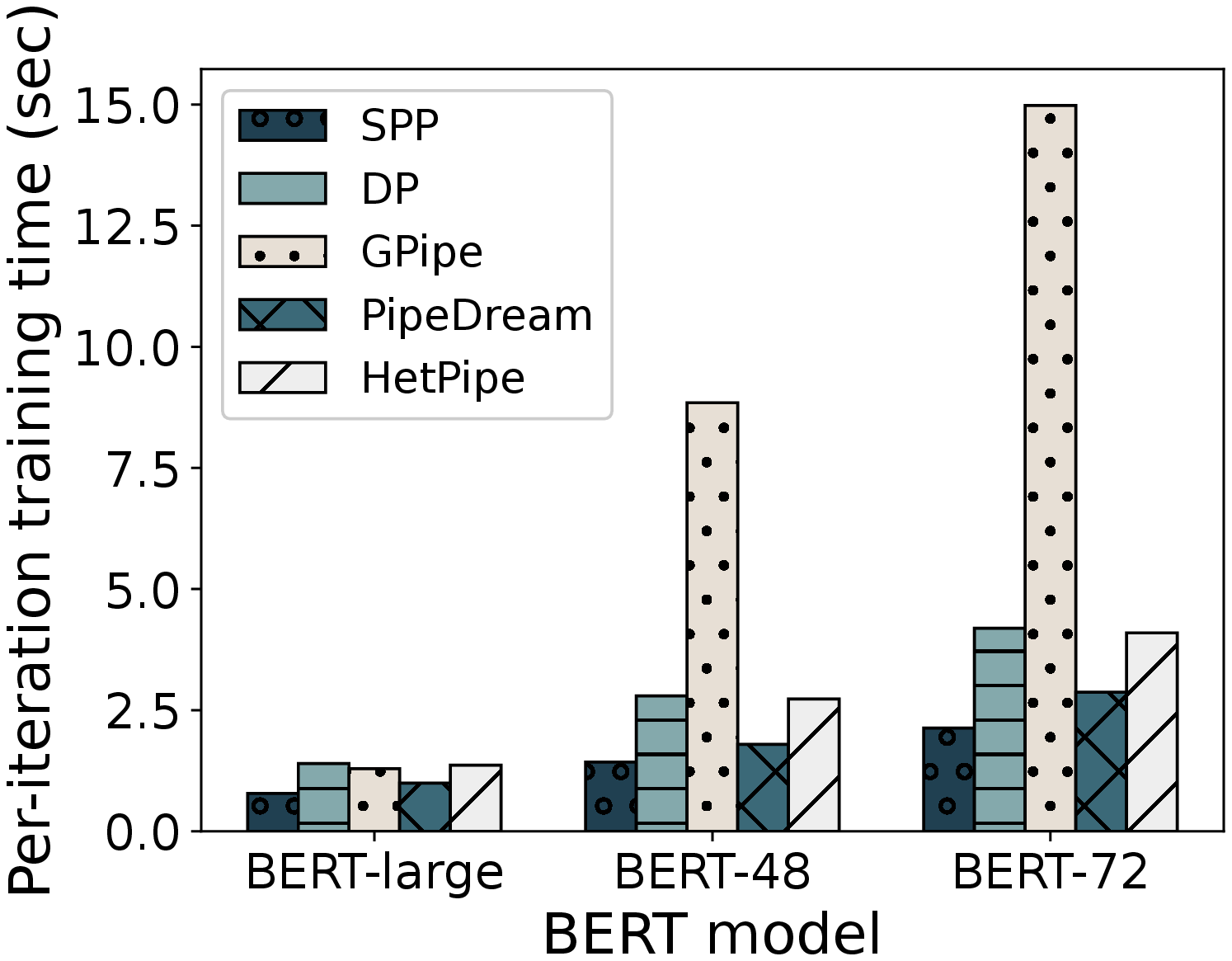}
\caption{Training time: BERT with different \# of layers}
\label{diff_bert_model}
\end{minipage}
\begin{minipage}[t]{0.45\columnwidth}
\includegraphics[width=\textwidth]{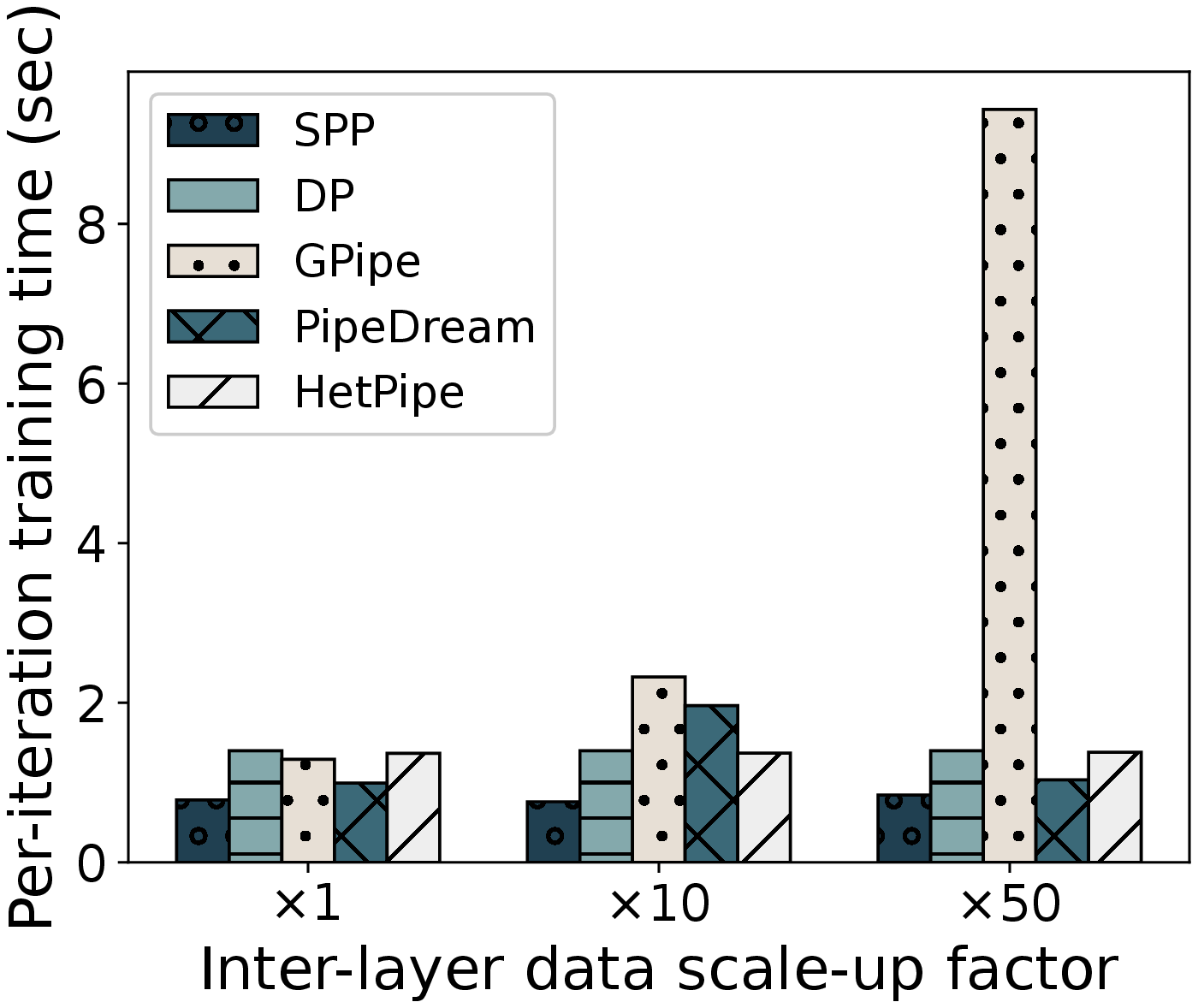}
\caption{Training time: different inter-layer data sizes}
\label{diff_comm_size}
\end{minipage}
\begin{minipage}[t]{0.46\columnwidth}
\includegraphics[width=1.1\textwidth]{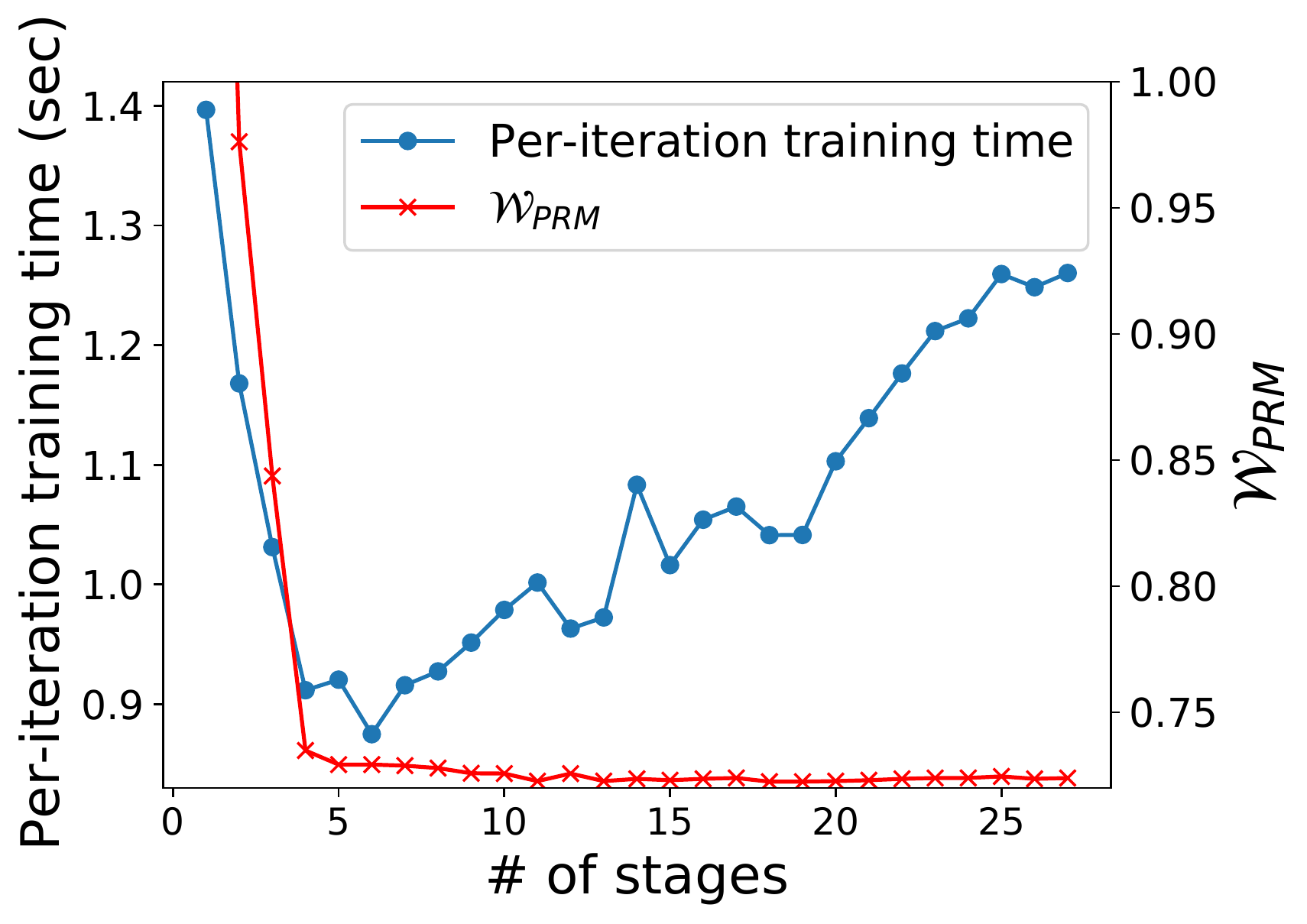}
\caption{\revise{Training time \& $\mathcal{W}_{PRM}$: different \# of stages}}
\label{diff_stage_num}
\end{minipage}
\end{figure}
\setlength{\textfloatsep}{0pt}

\vspace{-3mm}
\subsection{Trace-driven Simulation}
\vspace{-1mm}

\noindent\textbf{Settings.} By default, we simulate training of BERT-large (27-layers including 24 transformer layers) with 32 microbatches and a microbatch size of 6 on 8 servers, each equipped with 4 GPUs. We drive our simulation using profiled data 
collected by running the DNN on a V100 GPU. 
3 servers have intra-server bandwidth between [96, 128] Gbps (representing PCIe links~\cite{tallent2017evaluating}), and the other 5 servers 
[160, 200] Gbps (representing NVLink connections~\cite{amaral2017topology}). 
By default, inter-server bandwidth is set within [32, 40] Gbps to emulate an RDMA network~\cite{lu2018multi}. 


\subsubsection{Different numbers of microbatches} 
Fig.~\ref{diff_batch} shows that {\textit{SPP}} achieves significant training speed-up compared with the four baselines at different microbatch numbers ($M$). 
We also observed (figure omitted due to space limit) that higher GPU utilization is achieved with {\textit{SPP}} when $M$ is larger, implying the diminishing gap between {\textit{SPP}} and the optimal solution (which maximally utilizes GPUs for the best training speed). This is consistent with Theorem~\ref{th_approx_ratio}: as $M$ increases, the approximation ratio becomes smaller (\ie, better).

\subsubsection{Different inter-server bandwidth levels} We emulate three types of inter-server networks with low, medium and high bandwidth, respectively.
Fig.~\ref{diff_band} 
shows that per-iteration training time of {\textit{SPP}}, GPipe and PipeDream is stable 
at different bandwidth levels, 
while performance of DP and HetPipe drops dramatically at small bandwidth. This is because the former three 
overlap most inter-stage communication and AllReduce operations with computation, achieving higher GPU utilization even with small communication bandwidth. 
DP and HetPipe require an AllReduce operation over inter-server connections at the end of each training iteration, which incurs large communication time with small bandwidth.

\vspace{-1mm}
\subsubsection{Different inter-GPU connectivity} We 
next vary the number of available GPUs on servers and the server number. In Fig.~\ref{diff_inter_connectivity}, 
$[6\times2, 3\times4,1\times 8]$ represents training the model over 6 servers each with 2 GPUs, three 4-GPU servers and one 8-GPU server. 
{\textit{SPP}} achieves the best performance in all inter-GPU connection topologies. 
\subsubsection{Different numbers of layers} Fig.~\ref{diff_bert_model} compares the training performance 
of BERT-large, BERT-48 (48 Transformer layers) and BERT-72 (72 Transformer layers). As the model size increases, it is more difficult to obtain optimal model partition and device mapping solution. However, 
 performance of {\textit{SPP}} and PipeDream remains quite stable, 
with {\textit{SPP}} outperforming PipeDream by more than 20\% on the three models.
\subsubsection{Different inter-layer data sizes} We investigate the impact of activation sizes which influence inter-stage communication time, 
by scaling the 
activation data in BERT-large by different factors. 
Fig.~\ref{diff_comm_size} shows that the per-iteration training time with {\textit{SPP}} remains similar with the increase of activation sizes, due to 
the excellent communication and computation overlap it achieves. GPipe 
 tends to partition the model into more stages, resulting in more inter-stage communication time. 
\subsubsection{Different numbers of stages} Lemma~\ref{lemma_ss} gives that the performance of our algorithm is related to: (1) the number of stages $|\mathcal{S}|$, and (2) $\mathcal{W}_{PRM}$, the maximum time to process all microbatches on a single stage or communication channel. While PipeDream only aims at minimizing $\mathcal{W}_{PRM}$, {\textit{SPP}} strikes a balance between the two factors. In Fig.~\ref{diff_stage_num}, 
$\mathcal{W}_{PRM}$ first decreases when the model is partitioned into more stages, and becomes stable starting from the stage number of 4. The main reason is that for training BERT-large (with 24 uniform Transformer layers) on 32 GPUs, the per-stage training time with $4$ stages is already quite close to the optimal per-stage training time.\footnote{With 4 stages, we roughly have 6 layers per stage and each stage replicated to 8 GPUs; per-stage time is 6p/8 (p denotes per-layer computation time) plus AllReduce time. Optimal per-stage training time is lower bounded by 24p/32.} 
The 
training time first decreases as $\mathcal{W}_{PRM}$ drops,  and then increases when $\mathcal{W}_{PRM}$  stabilizes and $|\mathcal{S}|$ becomes the dominant factor, 
which is consistent with Lemma~\ref{lemma_ss}. 
This indicates that only minimizing $\mathcal{W}_{PRM}$ does not yield the best solution. 
{\textit{SPP}} strategically selects the 6-stage partition solution to minimize per-iteration training time.

\vspace{-2mm}
\section{Conclusion}
\label{sec::conclusion}
This paper designs efficient algorithms for expediting synchronous pipeline training of DNNs over arbitrary inter-GPU connectivity. We partition a given DNN, replicate and distribute the partitions over available GPUs, 
and design an efficient scheduler to order pipeline execution of microbatches over partitioned stages on different GPUs, minimizing the training time. 
Our comparative experiments on two GPU testbeds prove that our design outperforms state-of-the-art approaches up to 157\%. Trace-driven simulations further 
show our algorithms' superiority under various settings.


\appendices

\section{Pipeline Partition, Replication and Mapping Algorithm}
\label{appendix_alg_dp}
The pipeline partition, replication and mapping algorithm (PRM) is given in Alg.~\ref{alg_dp} 
If $l$ or $i$ is less than $\xi$, PRM terminates immediately as there is no feasible solution (lines 1-3). If $\xi$ equals 1 and $r$ is $i$, indicating there is only one stage, PRM groups the first $l$ layers as a single stage and replicates it over $\{v_1, \ldots, v_i\}$ (lines 4-12). Otherwise, we use dynamic programming to compute the optimal partition and mapping 
(lines 13-25). 

\begin{algorithm}[!t]

\caption{Pipeline Partition, Replication and Mapping - \textbf{\textit{PRM}}}
\label{alg_dp}

\renewcommand{\algorithmicrequire}{\textbf{Input:}}
\renewcommand{\algorithmicensure}{\textbf{Output:}}

\begin{algorithmic}[1]

\REQUIRE $G(\mathcal{V}, \mathcal{E}), \{v_1, v_2, \ldots, v_{N}\}, \mathcal{D}, l, n, \xi, r$
\ENSURE $W(l, \xi, r, n), \mathcal{S}, \mathcal{F}$
\IF {$l < \xi$ or $n < \xi$}
	\STATE Return $\textbf{INF}, \text{None}, \text{None}$
\ENDIF
\IF {$\xi = 1$ and $r = n$}
		\STATE $W(l, \xi, r, n) \leftarrow M\sum\limits_{i = 1}^{l}\frac{p^f_i + p^b_i}{n} + A_{1\rightarrow l}(v_{1} \rightarrow v_{n})$
		\STATE $s \leftarrow \{1, \ldots, l\}$
		\STATE $\mathcal{F}(s) \leftarrow \{v_1, \ldots, v_n\}$
\ELSE
	\IF {$\xi = 1$ or $r = n$}
		\STATE Return $\textbf{INF}, \text{None}, \text{None}$
	\ENDIF
\ENDIF
\STATE $min\_max\_time = \textbf{INF}, $
\FOR {$l' \in \{1, 2, \ldots, l-1\}$}
	\FOR {$r' \in \{1, 2, \ldots, n-r\}$}
		\STATE $W(l', {\xi-1}, r', {n-r}), \mathcal{S'}, \mathcal{F}\leftarrow \text{\textbf{\textit{PRM}}}(G, \{v_1, \ldots, v_n\}, \mathcal{D}, l', n-r, \xi-1, r')$
		\STATE Set:
		{\small
			\begin{multline*}
				max\_time \leftarrow \max\{W(l', \xi-1, r', n-r), \\M\frac{d^f_{l', l'+1} + d^b_{l'+1, l'}}{rr'b_{rr'}}, \\M\frac{\sum\limits_{i = l'+1}^{l}(p^f_i + p^b_i)}{r} + A_{l'+1\rightarrow l}(v_{n-r+1} \rightarrow v_{n})\}
			\end{multline*}}
		\IF {$min\_max\_time > max\_time$}
			\STATE $min\_max\_time \leftarrow max\_time$
			\STATE $s \leftarrow \{l'+1, l'+2, \ldots, l\}$
			\STATE $\mathcal{S} \leftarrow \mathcal{S'}\cup \{s\}$
			\STATE Cancel previous mapping, and Set $\mathcal{F}(s) \leftarrow \{v_{k+1}, \ldots, v_{n}\}$
		\ENDIF
	\ENDFOR
\ENDFOR
\STATE Return $min\_max\_time, \mathcal{S}, \mathcal{F}$
\end{algorithmic}
\end{algorithm}

\newpage
\bibliographystyle{IEEEtran}
\bibliography{./main}

\begin{thebibliography}{10}
\providecommand{\url}[1]{#1}
\csname url@samestyle\endcsname
\providecommand{\newblock}{\relax}
\providecommand{\bibinfo}[2]{#2}
\providecommand{\BIBentrySTDinterwordspacing}{\spaceskip=0pt\relax}
\providecommand{\BIBentryALTinterwordstretchfactor}{4}
\providecommand{\BIBentryALTinterwordspacing}{\spaceskip=\fontdimen2\font plus
\BIBentryALTinterwordstretchfactor\fontdimen3\font minus
  \fontdimen4\font\relax}
\providecommand{\BIBforeignlanguage}[2]{{%
\expandafter\ifx\csname l@#1\endcsname\relax
\typeout{** WARNING: IEEEtran.bst: No hyphenation pattern has been}%
\typeout{** loaded for the language `#1'. Using the pattern for}%
\typeout{** the default language instead.}%
\else
\language=\csname l@#1\endcsname
\fi
#2}}
\providecommand{\BIBdecl}{\relax}
\BIBdecl

\bibitem{devlin2018bert}
J.~Devlin, M.-W. Chang, K.~Lee, and K.~Toutanova, ``{{{BERT}}: Pre-Training of
  Deep Bidirectional Transformers for Language Understanding},'' \emph{arXiv
  preprint arXiv:1810.04805}, 2018.

\bibitem{he2016deep}
K.~He, X.~Zhang, S.~Ren, and J.~Sun, ``{Deep Residual Learning for Image
  Recognition},'' in \emph{Proc. of IEEE CVPR}, 2016.

\bibitem{cully2015robots}
A.~Cully, J.~Clune, D.~Tarapore, and J.-B. Mouret, ``{Robots That Can Adapt
  like Animals},'' \emph{Nature}, vol. 521, no. 7553, pp. 503--507, 2015.

\bibitem{abadi2016tensorflow}
M.~Abadi, P.~Barham, J.~Chen, Z.~Chen, A.~Davis, J.~Dean, M.~Devin,
  S.~Ghemawat, G.~Irving, M.~Isard \emph{et~al.}, ``{{TensorFlow}: A System for
  Large-Scale Machine Learning},'' in \emph{Proc. of USENIX OSDI}, 2016.

\bibitem{adam2019pytorch}
A.~Paszke, S.~Gross, F.~Massa, A.~Lerer, J.~Bradbury, G.~Chanan, T.~Killeen,
  Z.~Lin, N.~Gimelshein, L.~Antiga, A.~Desmaison, A.~Kopf, E.~Yang, Z.~DeVito,
  M.~Raison, A.~Tejani, S.~Chilamkurthy, B.~Steiner, L.~Fang, J.~Bai, and
  S.~Chintala, ``{{{PyTorch}}: An Imperative Style, High-Performance Deep
  Learning Library},'' in \emph{Proc. of NeurIPS}, 2019, pp. 8024--8035.

\bibitem{li2014scaling}
M.~Li, D.~G. Andersen, J.~W. Park, A.~J. Smola, A.~Ahmed, V.~Josifovski,
  J.~Long, E.~J. Shekita, and B.-Y. Su, ``{Scaling Distributed Machine Learning
  with the Parameter Server},'' in \emph{Proc. of USENIX OSDI}, 2014.

\bibitem{sergeev2018horovod}
A.~Sergeev and M.~Del~Balso, ``{Horovod: Fast and Easy Distributed Deep
  Learning in TensorFlow},'' \emph{arXiv preprint arXiv:1802.05799}, 2018.

\bibitem{shoeybi2019megatron}
M.~Shoeybi, M.~Patwary, R.~Puri, P.~LeGresley, J.~Casper, and B.~Catanzaro,
  ``{Megatron-LM: Training Multi-Billion Parameter Language Models Using GPU
  Model Parallelism},'' \emph{arXiv preprint arXiv:1909.08053}, 2019.

\bibitem{harlap2016addressing}
A.~Harlap, H.~Cui, W.~Dai, J.~Wei, G.~R. Ganger, P.~B. Gibbons, G.~A. Gibson,
  and E.~P. Xing, ``{Addressing the Straggler Problem for Iterative Convergent
  Parallel ML},'' in \emph{Proc. of ACM SoCC}, 2016.

\bibitem{harlap2018pipedream}
D.~Narayanan, A.~Harlap, A.~Phanishayee, V.~Seshadri, N.~R. Devanur, G.~R.
  Ganger, P.~B. Gibbons, and M.~Zaharia, ``{PipeDream: Generalized Pipeline
  Parallelism for DNN Training},'' in \emph{Proc. of ACM SOSP}, 2019.

\bibitem{geng2019elasticpipe}
J.~Geng, D.~Li, and S.~Wang, ``{{ElasticPipe}: An Efficient and Dynamic
  Model-Parallel Solution to {Dnn} Training},'' in \emph{Proc. of the 10th
  Workshop on Scientific Cloud Computing}, 2019.

\bibitem{Park2020hetpipe}
J.~H. Park, G.~Yun, C.~M. Yi, N.~T. Nguyen, S.~Lee, J.~Choi, S.~H. Noh, and
  Y.~ri~Choi, ``{{HetPipe}: Enabling Large {DNN} Training on (Whimpy)
  Heterogeneous {{GPU}} Clusters through Integration of Pipelined Model
  Parallelism and Data Parallelism},'' in \emph{Proc. of USENIX ATC}, 2020.

\bibitem{pmlr-v139-narayanan21a}
D.~Narayanan, A.~Phanishayee, K.~Shi, X.~Chen, and M.~Zaharia,
  ``{Memory-Efficient Pipeline-Parallel DNN Training},'' in \emph{Proc. of
  ICML}, 2021.

\bibitem{ho2013more}
Q.~Ho, J.~Cipar, H.~Cui, S.~Lee, J.~K. Kim, P.~B. Gibbons, G.~A. Gibson,
  G.~Ganger, and E.~P. Xing, ``{More Effective Distributed ML via a Stale
  Synchronous Parallel Parameter Server},'' in \emph{Proc. of NeurIPS}, 2013.

\bibitem{huang2019gpipe}
Y.~Huang, Y.~Cheng, A.~Bapna, O.~Firat, D.~Chen, M.~Chen, H.~Lee, J.~Ngiam,
  Q.~V. Le, Y.~Wu \emph{et~al.}, ``{GPipe: Efficient Training of Giant Neural
  Networks Using Pipeline Parallelism},'' in \emph{Proc. of NeurIPS}, 2019.

\bibitem{fan2021dapple}
S.~Fan, Y.~Rong, C.~Meng, Z.~Cao, S.~Wang, Z.~Zheng, C.~Wu, G.~Long, J.~Yang,
  L.~Xia \emph{et~al.}, ``{{{DAPPLE}}: A Pipelined Data Parallel Approach for
  Training Large Models},'' in \emph{Proc.~of ACM PPoPP}, 2021, pp. 431--445.

\bibitem{dgx1}
\emph{NVIDIA DGX-1}, \url{https://www.nvidia.com/en-us/data-center/dgx-1/}.

\bibitem{wang2013gpu}
H.~Wang, S.~Potluri, D.~Bureddy, C.~Rosales, and D.~K. Panda, ``{{GPU-aware}
  {MPI on RDMA-enabled} Clusters: Design, Implementation and Evaluation},''
  \emph{IEEE Transactions on Parallel and Distributed Systems}, vol.~25,
  no.~10, pp. 2595--2605, 2013.

\bibitem{goldberg2001better}
L.~A. Goldberg, M.~Paterson, A.~Srinivasan, and E.~Sweedyk, ``{Better
  Approximation Guarantees for Job-Shop Scheduling},'' \emph{SIAM Journal on
  Discrete Mathematics}, vol.~14, no.~1, pp. 67--92, 2001.

\bibitem{goodfellow2016deep}
I.~Goodfellow, Y.~Bengio, and A.~Courville, \emph{{Deep Learning}}.\hskip 1em
  plus 0.5em minus 0.4em\relax MIT press, 2016.

\bibitem{mirhoseini2017device}
A.~Mirhoseini, H.~Pham, Q.~V. Le, B.~Steiner, R.~Larsen, Y.~Zhou, N.~Kumar,
  M.~Norouzi, S.~Bengio, and J.~Dean, ``{Device Placement Optimization with
  Reinforcement Learning},'' in \emph{Proc. of ICML}.\hskip 1em plus 0.5em
  minus 0.4em\relax PMLR, 2017, pp. 2430--2439.

\bibitem{addanki2018placeto}
R.~Addanki, S.~B. Venkatakrishnan, S.~Gupta, H.~Mao, and M.~Alizadeh,
  ``{Placeto: Efficient Progressive Device Placement Optimization},'' in
  \emph{NIPS Machine Learning for Systems Workshop}, 2018.

\bibitem{yi2020optimizing}
X.~Yi, S.~Zhang, Z.~Luo, G.~Long, L.~Diao, C.~Wu, Z.~Zheng, J.~Yang, and
  W.~Lin, ``{Optimizing Distributed Training Deployment in Heterogeneous GPU
  Clusters},'' in \emph{Proc. of International Conference on emerging
  Networking EXperiments and Technologies}, 2020, pp. 93--107.

\bibitem{wu2020stanza}
X.~{Wu}, H.~{Xu}, B.~{Li}, and Y.~{Xiong}, ``{Stanza: Layer Separation for
  Distributed Training in Deep Learning},'' \emph{IEEE Transactions on Services
  Computing}, pp. 1--1, 2020.

\bibitem{yi2020fast}
X.~Yi, Z.~Luo, C.~Meng, M.~Wang, G.~Long, C.~Wu, J.~Yang, and W.~Lin, ``{Fast
  Training of Deep Learning Models over Multiple GPUs},'' in \emph{Proc. of the
  21st International Middleware Conference}, 2020, pp. 105--118.

\bibitem{chen2018efficient}
C.-C. Chen, C.-L. Yang, and H.-Y. Cheng, ``{Efficient and Robust Parallel DNN
  Training Through Model Parallelism on Multi-GPU Platform},'' \emph{arXiv
  preprint arXiv:1809.02839}, 2018.

\bibitem{colin2019theoretical}
I.~Colin, L.~Dos~Santos, and K.~Scaman, ``{Theoretical Limits of Pipeline
  Parallel Optimization and Application to Distributed Deep Learning},'' in
  \emph{Proc. of NeurIPS}, 2019.

\bibitem{allreduce_time}
\emph{Performance reported by NCCL tests},
  \url{https://github.com/NVIDIA/nccl-tests/blob/master/doc/PERFORMANCE.md}.

\bibitem{benoit2008mapping}
A.~Benoit and Y.~Robert, ``{Mapping Pipeline Skeletons onto Heterogeneous
  Platforms},'' \emph{Journal of Parallel and Distributed Computing}, vol.~68,
  no.~6, pp. 790--808, 2008.

\bibitem{stoer1997simple}
M.~Stoer and F.~Wagner, ``{A Simple Min-Cut Algorithm},'' \emph{Journal of the
  ACM}, vol.~44, no.~4, pp. 585--591, 1997.

\bibitem{NCCL}
\emph{NVIDIA Collective Communication Library},
  \url{https://github.com/NVIDIA/nccl}.

\bibitem{deng2009imagenet}
J.~Deng, W.~Dong, R.~Socher, L.-J. Li, K.~Li, and L.~Fei-Fei, ``{{ImageNet}: A
  Large-Scale Hierarchical Image Database},'' in \emph{Proc. of IEEE CVPR},
  2009.

\bibitem{rajpurkar2018know}
P.~Rajpurkar, R.~Jia, and P.~Liang, ``{Know What You Don't Know: Unanswerable
  Questions for {SQuAD}},'' \emph{arXiv preprint arXiv:1806.03822}, 2018.

\bibitem{simonyan2014very}
K.~Simonyan and A.~Zisserman, ``{Very Deep Convolutional Networks for
  Large-Scale Image Recognition},'' \emph{arXiv preprint arXiv:1409.1556},
  2014.

\bibitem{szegedy2016rethinking}
C.~Szegedy, V.~Vanhoucke, S.~Ioffe, J.~Shlens, and Z.~Wojna, ``{Rethinking the
  Inception Architecture for Computer Vision},'' in \emph{Proc. of IEEE CVPR},
  2016.

\bibitem{vaswani2017attention}
A.~Vaswani, N.~Shazeer, N.~Parmar, J.~Uszkoreit, L.~Jones, A.~N. Gomez,
  {\L}.~Kaiser, and I.~Polosukhin, ``{Attention Is All You Need},'' in
  \emph{Proc. of NeurIPS}, 2017.

\bibitem{yang2019xlnet}
Z.~Yang, Z.~Dai, Y.~Yang, J.~Carbonell, R.~R. Salakhutdinov, and Q.~V. Le,
  ``{{XLNet}: Generalized Autoregressive Pretraining for Language
  Understanding},'' in \emph{Proc. of NeurIPS}, 2019.

\bibitem{tallent2017evaluating}
N.~R. Tallent, N.~A. Gawande, C.~Siegel, A.~Vishnu, and A.~Hoisie,
  ``{Evaluating on-Node GPU Interconnects for Deep Learning Workloads},'' in
  \emph{Proc. of International Workshop on Performance Modeling, Benchmarking
  and Simulation of High Performance Computer Systems}, 2017.

\bibitem{amaral2017topology}
M.~Amaral, J.~Polo, D.~Carrera, S.~Seelam, and M.~Steinder, ``{Topology-Aware
  GPU Scheduling for Learning Workloads in Cloud Environments},'' in
  \emph{Proc. of SC}, 2017.

\bibitem{lu2018multi}
Y.~Lu, G.~Chen, B.~Li, K.~Tan, Y.~Xiong, P.~Cheng, J.~Zhang, E.~Chen, and
  T.~Moscibroda, ``{Multi-Path Transport for RDMA in Datacenters},'' in
  \emph{Proc. of USENIX NSDI}, 2018.

\end{thebibliography}

\end{document}